\documentclass[10pt, notitlepage]{article}

\usepackage{a4}

\usepackage{xcolor}

\definecolor{TUMblue}{RGB}{0, 101, 189}
\definecolor{TUMlightblue}{RGB}{100,160,200}
\definecolor{TUMgreen}{RGB}{162,173,0}
\definecolor{TUMorange}{RGB}{227,114,034}
\definecolor{TUMivory}{RGB}{218,215,203}

\usepackage{natbib}

\usepackage{hyperref}
\hypersetup{
	colorlinks=true,
	linkcolor=TUMblue,
	citecolor=TUMblue,
	filecolor=TUMblue,
	urlcolor=TUMblue
}

\usepackage{etoolbox}
\makeatletter

\pretocmd{\NAT@citex}{%
	\let\NAT@hyper@\NAT@hyper@citex
	\def\NAT@postnote{#2}%
	\setcounter{NAT@total@cites}{0}%
	\setcounter{NAT@count@cites}{0}%
	\forcsvlist{\stepcounter{NAT@total@cites}\@gobble}{#3}}{}{}
\newcounter{NAT@total@cites}
\newcounter{NAT@count@cites}
\def\NAT@postnote{}

\def\NAT@hyper@citex#1{%
	\stepcounter{NAT@count@cites}%
	\hyper@natlinkstart{\@citeb\@extra@b@citeb}#1%
	\ifnumequal{\value{NAT@count@cites}}{\value{NAT@total@cites}}
	{\ifNAT@swa\else\if*\NAT@postnote*\else%
		\NAT@cmt\NAT@postnote\global\def\NAT@postnote{}\fi\fi}{}%
	\ifNAT@swa\else\if\relax\NAT@date\relax
	\else\NAT@@close\global\let\NAT@nm\@empty\fi\fi
	\hyper@natlinkend}
\renewcommand\hyper@natlinkbreak[2]{#1}

\patchcmd{\NAT@citex}
{\ifNAT@swa\else\if*#2*\else\NAT@cmt#2\fi
	\if\relax\NAT@date\relax\else\NAT@@close\fi\fi}{}{}{}
\patchcmd{\NAT@citex}
{\if\relax\NAT@date\relax\NAT@def@citea\else\NAT@def@citea@close\fi}
{\if\relax\NAT@date\relax\NAT@def@citea\else\NAT@def@citea@space\fi}{}{}

\makeatother

\usepackage{amssymb, graphicx}
\usepackage{subfigure}
\usepackage{amsmath}
\usepackage{amsthm}
\usepackage{undertilde}
\usepackage{verbatim}
\usepackage{bbm}
\usepackage{ dsfont }
\usepackage{geometry}
\usepackage{pdflscape}
\usepackage{multirow}
\usepackage{aliascnt}

\newcommand{\mynewtheorem}[2]{
	\newaliascnt{#1}{dummy}
	\newtheorem{#1}[#1]{#2}
	\aliascntresetthe{#1}
	\expandafter\def\csname #1autorefname\endcsname{#2}
}

\theoremstyle{definition}
\mynewtheorem{thm}{Theorem}
\mynewtheorem{defi}{Definition}
\mynewtheorem{lem}{Lemma}
\mynewtheorem{cor}{Corollary}
\mynewtheorem{prop}{Proposition}
\mynewtheorem{exa}{Example}
\mynewtheorem{alg}{Algorithm}
\mynewtheorem{rem}{Remark}
\mynewtheorem{bsp}{Example}

\def\equationautorefname~#1\null{Equation~(#1)\null}
\newcommand{\aref}[1]{\hyperref[#1]{Appendix~\ref{#1}}}

\usepackage{sectsty}
\allsectionsfont{\sffamily}

\usepackage{url}

\usepackage{caption}
\usepackage{footnote}

\newcommand{\Ebb}{\mathbb{E}}

\newcommand{\Nbb}{\mathbb{N}}

\newcommand{\Rbb}{\mathbb{R}}

\newcommand{\gammab}{\mathbf{\boldsymbol{\gamma}}}

\newcommand{\omegab}{\boldsymbol{\omega}}

\newcommand{\rb}{\mathbf{r}}

\newcommand{\sigmab}{\boldsymbol{\sigma}}

\newcommand{\ub}{\mathbf{u}}
\newcommand{\Ub}{\mathbf{U}}

\newcommand{\vb}{\mathbf{v}}

\newcommand{\wb}{\mathbf{w}}
\newcommand{\Xb}{\mathbf{X}}
\newcommand{\xb}{\mathbf{x}}
\newcommand{\xib}{\boldsymbol{\xi}}

\newcommand{\Cc}{\mathcal{C}}
\newcommand{\Dc}{\mathcal{D}}
\newcommand{\Fc}{\mathcal{F}}
\newcommand{\Gc}{\mathcal{G}}

\newcommand{\Jc}{\mathcal{J}}

\newcommand{\Nc}{\mathcal{N}}

\newcommand{\Rc}{\mathcal{R}}
\newcommand{\Uc}{\mathcal{U}}

\newcommand{\Wc}{\mathcal{W}}

\newcommand{\be}{\begin{equation}}
	\newcommand{\ee}{\end{equation}}

\newcommand{\eps}{\varepsilon}

\newcommand{\Ast}{\mathop{\scalebox{1.5}{\raisebox{-0.2ex}{$\ast$}}}}%

\DeclareMathOperator{\aKL}{aKL}
\DeclareMathOperator{\argmax}{arg\,max}

\DeclareMathOperator{\dKL}{dKL}
\DeclareMathOperator{\diag}{diag}

\DeclareMathOperator{\KL}{KL}
\DeclareMathOperator{\MCKL}{MCKL}

\DeclareMathOperator{\sdKL}{sdKL}

\DeclareMathOperator{\tr}{tr}

\newcommand*\diff{\mathop{}\!\mathrm{d}}

\begin{document} 
	
	{
		\renewcommand*{\thefootnote}{\fnsymbol{footnote}}
		\title{\textbf{\sffamily Model distances for vine copulas in high dimensions
			}}
			\date{\small \today}
			\author{Matthias Killiches\footnote{Zentrum Mathematik, Technische Universit\"at M\"unchen, Boltzmannstra\ss e 3, 85748 Garching, Germany} \footnote{Corresponding author, email: \texttt{matthias.killiches@tum.de}.} , Daniel Kraus$^*$ and Claudia Czado$^*$}
		
			\maketitle
			
			\begin{abstract}
				Vine copulas are a flexible class of dependence models
				consisting of bivariate building blocks and have proven to be
				particularly useful in high dimensions. Classical model distance measures require	multivariate integration and thus suffer from the curse of
				dimensionality. In this paper we provide numerically tractable methods
				to measure the distance between two vine copulas even in high
				dimensions. For this purpose, we consecutively develop three new distance measures based on the Kullback-Leibler distance, using the result that it can be expressed as the sum over expectations of KL distances between univariate conditional densities, which can be easily obtained for vine copulas. To reduce numerical calculations we approximate these expectations on adequately designed grids, outperforming Monte Carlo-integration with respect to computational time. In numerous examples and applications we illustrate the strengths and weaknesses of the developed distance measures.\\
				
				\noindent \textsf{Keywords:} \textit{Vine copulas; model distances; Kullback-Leibler; Monte Carlo-integration.}
			\end{abstract}
		}

\section{Introduction}\label{sec:intro}
In the course of growing data sets and increasing computing power statistical data analysis has considerably developed within the last decade. The necessity of proper dependence modeling has become evident at least since the financial crisis of 2007. Using vine copulas is a popular option to approach this task. \cite{bedford2002vines} described how multivariate distributions can be sequentially decomposed into bivariate building blocks via conditioning. Since the seminal paper of \cite{aasczado}, which developed statistical inference for this method, many aspects of vines have been studied: \cite{dissmann2013selecting} provide a sequential estimation algorithm for vines, \cite{panagiotelis2012pair} treat vine copulas for discrete data and \cite{nagler2015evading} examine non-parametric vine copulas. Further, there have been various applications to data from many fields such as finance \citep{maya2015latin,kraus2015d}, sociology \citep{cooke2015vine} or hydrology \citep{killiches2015maxima}. The advantage of these models is that they are flexible and numerically tractable even in high dimensions.

Since it is interesting in many cases to determine how much two models differ, some authors like \cite{stoeber2013simplified} and \cite{schepsmeier2015efficient} use the \textit{Kullback-Leibler distance} \citep{kullback1951information} as a model distance between vines. However, all popular distance measures require multivariate integration, which is why they can only deal with up to three- or four-dimensional models in a reasonable amount of time.

In this paper we will address the question of how to measure the distance between two vine copulas even for high dimensions. For this purpose, we develop methods based on the Kullback-Leibler (KL) distance, where we use the fact that it can be expressed as the sum over expectations of KL distances between univariate conditional densities. By cleverly approximating these expectations in different ways, we introduce three new distance measures with varying focuses. The \textit{approximate Kullback-Leibler distance} (aKL) aims to approximate the true Kullback-Leibler distance via structured Monte Carlo integration and is a computationally tractable distance measure in up to five dimensions. The \textit{diagonal Kullback-Leibler distance} (dKL) focuses on the distance between two vine copulas on specific conditioning vectors, namely those lying on certain diagonals in the space. We show that even though the resulting distance measure does not approximate the KL distance in a classical sense, it still reproduces its qualitative behavior quite well. While this way of measuring distances between vines is fast in up to ten dimensions, we still have to reduce the number of evaluation points in order to get a numerically tractable distance measure for dimensions 30 and higher. By concentrating on only one specific diagonal we achieve this, defining the \textit{single diagonal Kullback-Leibler distance} (sdKL). In numerous examples and applications we illustrate that the proposed methods are valid distance measures and outperform benchmark approaches like Monte Carlo integration regarding computational time.

The paper is organized as follows: \autoref{sec:vines} introduces vine copulas and basic properties. In \autoref{sec:modeldist} we develop the above mentioned model distances for vines and compare their performances in various settings. \autoref{sec:conclusion} concludes with a summary and an outlook to ongoing research.

\section{Vine Copulas}\label{sec:vines}

A \emph{copula} $C\colon [0,1]^d\to [0,1]$ is a $d$-dimensional distribution function on $[0,1]^d$ with uniformly distributed margins. Since the publication of \cite{Sklar}, copulas have gained more and more interest and have been a frequent subject in many areas of probabilistic and statistical research. Sklar's Theorem states that for every joint distribution function $F\colon \Rbb^d \to [0,1]$ of a $d$-dimensional random variable $(X_1,\ldots,X_d)'$ with univariate marginal distribution functions $F_j$, $j=1,\ldots,d$, there exists a copula $C$ such that
\begin{equation}\label{eq:sklar}
F(x_1,\ldots,x_d)=C\left(F_1(x_1),\ldots, F_d(x_d)\right).
\end{equation}
This copula $C$ is unique if all $X_j$ are continuous random variables. Further, if the so-called \emph{copula density}
\[
c(u_1,\ldots,u_d):=\frac{\partial^d}{\partial u_1 \cdots \partial u_d}C(u_1,\ldots,u_d)
\]
exists, one has
\[
f(x_1,\ldots,x_d)=c\left(F_1(x_1),\ldots, F_d(x_d)\right)f_1(x_1)\cdots f_d(x_d),
\]
where $f_j$ are the marginal densities. In the following we will always assume absolute continuity of $C$ and the existence of $c$. \autoref{eq:sklar} can also be used to define a multivariate distribution by combining a copula $C$ and marginal distribution functions $F_j$. Thus, marginals and dependence structure can be modeled separately, as we can specify the copula $C$ independently of the marginal distributions. A thorough overview over copulas can be found in \cite{joe1997multivariate} and \cite{nelsen2006introduction}.

There are several multivariate parametric copula families, for example Gaussian, t, Gumbel, Clayton and Joe copulas. Being specified by a small number of parameters (usually 1 or 2), these models are rather inflexible in high dimensions. Therefore, \cite{bedford2002vines} suggested a method for constructing copula densities based on the combination of bivariate building blocks: \emph{vines}. The concept of vine copulas, also referred to as \textit{pair-copula constructions} (PCCs), was used by \cite{aasczado} to develop statistical inference methods.

As an example, a three-dimensional copula density $c$ of a random vector $(U_1,U_2,U_3)'$ with $U_j\sim \text{uniform}(0,1)$ can be decomposed by conditioning on $U_2=u_2$ and using $c_j(u_j)=1$:
\begin{equation} \label{eq:3dvinedensity}
\begin{split}
c(u_1,u_2,u_3)&=c_{13| 2}(u_1,u_3| u_2)\,c_2(u_2)\\
&\stackrel{Sklar}{=} c_{13;2}\left(C_{1| 2}(u_1| u_2),C_{3| 2}(u_3|u_2);u_2\right)c_{1|2}(u_1| u_2)\,c_{2| 3}(u_2| u_3)\\
&=c_{13;2}\left(C_{1|2}(u_1|u_2),C_{3|2}(u_3|u_2);u_2\right)c_{12}(u_1,u_2)\,c_{23}(u_2,u_3),
\end{split}
\end{equation}
where $c_{13| 2}(\,\cdot\,,\cdot\,|u_2)$ denotes the density of the conditional distribution of $(U_1,U_3)|U_2=u_2$, while $c_{13;2}(\,\cdot\,,\cdot\,;u_2)$ is the associated copula density. The distribution function of the conditional distribution of $U_j$ given $U_2=u_2$ is denoted by $C_{j|2}(\,\cdot\,| u_2)$, $j=1,3$. Hence, we have expressed the three-dimensional copula density as the product over three bivariate pair-copulas.

Of course, there are alternative decompositions since the choice of $U_2$ as conditioning variable was arbitrary. For example, we also could have conditioned on $U_1$ or $U_3$ such that
\[
c(u_1,u_2,u_3)=c_{23;1}\left(C_{2|1}(u_2|u_1),C_{3|1}(u_3|u_1);u_1\right)c_{12}(u_1,u_2)\,c_{13}(u_1,u_3),
\]
\[
c(u_1,u_2,u_3)=c_{12;3}\left(C_{1|3}(u_1|u_3),C_{2|3}(u_2|u_3);u_3\right)c_{13}(u_1,u_3)\,c_{23}(u_2,u_3).
\]
This way of decomposing copula densities into bivariate building blocks can be extended to arbitrary dimensions. \cite{morales2011count} show that in $d$ dimensions there are $\frac{d!}{2}\cdot 2^{\binom{d-2}{2}}$ possible vine decompositions. This flexibility and variety of choice can be of great advantage when it comes to modeling.

\cite{dissmann2013selecting} and \cite{stoeber2012} provide a method of how to store the structure of a vine copula decomposition in a lower triangular matrix $M=\left(m_{i,j}\right)_{i,j=1}^d$ with $m_{i,j}=0$ for $i<j$, a so-called \emph{vine structure matrix}.
\begin{defi}[Vine structure matrix]\label{defi:RVM}
	A lower-triangular matrix $M=\left(m_{i,j}\right)_{i,j=1}^d$ is called a \emph{vine structure matrix} if it has the following three properties:
	\begin{enumerate}
		\item The entries of a selected column appear in every column to the left of that column, i.e.\ for $1\leq i<j\leq d$ it holds $\left\lbrace m_{j,j},\ldots,m_{d,j}\right\rbrace \subseteq \left\lbrace m_{i,i},\ldots,m_{d,i}\right\rbrace$.
		\item The diagonal entry of a column does not appear in any column further to the right, i.e.\ $m_{i,i}\notin \left\lbrace m_{i+1,i+1},\ldots,m_{d,i+1}\right\rbrace$ for $i=1,\ldots,d-1$.
		\item For $i=1,\ldots,d-2$ and $k=i+1,\ldots,d$ there exists a $j>i$ such that
		\begin{align*}
		\left\lbrace m_{k,i}, \left\lbrace m_{k+1,i},\ldots, m_{d,i} \right\rbrace \right\rbrace & = \left\lbrace m_{j,j}, \left\lbrace m_{k+1,j},m_{k+2,j},\ldots, m_{d,j} \right\rbrace \right\rbrace \quad \text{or}\\
		\left\lbrace m_{k,i}, \left\lbrace m_{k+1,i},\ldots, m_{d,i} \right\rbrace \right\rbrace & = \left\lbrace m_{k+1,j}, \left\lbrace m_{j,j},m_{k+2,j},\ldots, m_{d,j} \right\rbrace \right\rbrace.
		\end{align*}
		
	\end{enumerate}
\end{defi}
The structure of the vine is encoded in the matrix as subsequently described: A pair-copula is determined by the two conditioned variables and a (possibly empty) set of conditioning variables (e.g.\ $c_{1,3;2}$ has conditioned variables $U_1$ and $U_3$ and conditioning variable $U_2$). For each entry in the structure matrix, the entry $m_{i,j}$ itself and the diagonal entry $m_{j,j}$ in the corresponding column form the indices of the two conditioned variables, while the indices of the conditioning variables are given by the entries $m_{i+1,j},\ldots,m_{d,j}$ in the corresponding column below the considered entry. The bivariate pair-copulas are evaluated at the conditional distribution functions of the distributions of each of the conditioned variables given the conditioning variables.

Expressed in formulas this means: In $d$ dimensions the entry $m_{i,j}$ ($i>j$) together with $m_{j,j}$ and $m_{i+1},\ldots,m_{d,j}$ stands for the copula density of the (conditional) distribution of $U_{m_{i,j}}$ and $U_{m_{j,j}}$ given $\left(U_{m_{i+1,j}},\ldots,U_{m_{d,j}}\right)'=\left(u_{m_{i+1,j}},\ldots,u_{m_{d,j}}\right)'$ evaluated at $C_{m_{i,j}|m_{i+1,j},\ldots,m_{d,j}}\left(u_{m_{i,j}}|u_{m_{i+1,j}},\ldots,u_{m_{d,j}}\right)$ and
$C_{m_{j,j}|m_{i+1,j},\ldots,m_{d,j}}\left(u_{m_{j,j}}|u_{m_{i+1,j}},\ldots,u_{m_{d,j}}\right)$, i.e.
\begin{multline*}
c_{m_{i,j},m_{j,j};m_{i+1,j},\ldots,m_{d,j}}\Big(C_{m_{i,j}|m_{i+1,j},\ldots,m_{d,j}}\big(u_{m_{i,j}}|u_{m_{i+1,j}},\ldots,u_{m_{d,j}}\big),\\ C_{m_{j,j}|m_{i+1,j},\ldots,m_{d,j}}\big(u_{m_{j,j}}|u_{m_{i+1,j}},\ldots,u_{m_{d,j}}\big);u_{m_{i+1,j}},\ldots,u_{m_{d,j}}\Big).
\end{multline*}
Taking the product over all $d(d-1)/2$ pair-copula expressions implied by the vine structure matrix yields the copula density $c$ \citep[see][]{dissmann2013selecting}:
\begin{equation}\label{eq:vinepdfdecomp}
\begin{split}
c(u_1,\ldots,u_d)&=\prod_{j=1}^{d-1}\prod_{k=j+1}^{d} c_{m_{k,j},m_{j,j};m_{k+1,j},\ldots,m_{d,j}}\big(C_{m_{k,j}|m_{k+1,j},\ldots,m_{d,j}}\left(u_{m_{k,j}}|u_{m_{k+1,j}},\ldots,u_{m_{d,j}}\right),\\ &C_{m_{j,j}|m_{k+1,j},\ldots,m_{d,j}}\left(u_{m_{j,j}}|u_{m_{k+1,j}},\ldots,u_{m_{d,j}}\right)\!;u_{m_{k+1,j}},\ldots,u_{m_{d,j}}\big).
\end{split}
\end{equation}
In our three-dimensional example (\autoref{eq:3dvinedensity}) the structure matrix looks as follows:
\[
M=\begin{pmatrix}
m_{1,1} & m_{1,2} & m_{1,3} \\
m_{2,1} & m_{2,2} & m_{2,3} \\
m_{3,1} & m_{3,2} & m_{3,3}
\end{pmatrix}=\begin{pmatrix}
1 & 0 & 0 \\
3 & 2 & 0 \\
2 & 3 & 3
\end{pmatrix}.
\]
The entries $m_{3,1}=2$ (together with $m_{1,1}=1$) and $m_{3,2}=3$ (together with $m_{2,2}=2$) in the last row represent $c_{1,2}(u_1,u_2)$ and $c_{2,3}(u_2,u_3)$, respectively. In both cases, the conditioning set is empty because the considered entries are the last ones in their columns. The entry $m_{2,1}$ (together with $m_{1,1}$ and $m_{3,1}$) encodes the expression $c_{1,3;2}\left(C_{1|2}(u_1|u_2),C_{3|2}(u_3|u_2);u_2\right)$ since the indices of the conditioned variables are given by $m_{2,1}=3$ and $m_{1,1}=1$ and the conditioning variable is $m_{3,1}=2$. Multiplying these three factors leads to the expression from \autoref{eq:3dvinedensity}. Note that there is not a unique way of encoding a given vine decomposition into a structure matrix. For instance, exchanging $m_{2,2}$ and $m_{3,2}$ in the above example yields the same vine decomposition.

Property 2 from \autoref{defi:RVM} implies that the diagonal of any vine structure matrix is a permutation of ${1\!:\!d}$, where we use the notation ${r\!:\!s}$ to describe the vector $\left( r,r+1,\ldots,s\right)' $ for $r\leq s$. In order to simplify notation, for the remainder of the paper we assume that the diagonal of a $d$-dimensional structure matrix is ${1\!:\!d}$. This assumption comes without any loss of generality since relabeling of the variables suffices to obtain the desired property.

The following \autoref{prop:column} states that for a vine copula with structure matrix $M$ the (univariate) conditional density $c_{j|(j+1):d}$ of $U_{j}\,|\left(U_{j+1},\ldots,U_d\right)'=\left(u_{j+1},\ldots,u_{d}\right)'$ can be calculated by taking the product over all pair-copula expressions corresponding to the entries in the $j$th column of $M$. A proof can be found in \autoref{App:column}.
\begin{prop}\label{prop:column}
	Let $\Ub=(U_1,\ldots,U_d)'$ be a random vector with vine copula density $c$ and corresponding structure matrix ${M=\left(m_{i,j}\right)_{i,j=1}^d}$. Then, for $j<d$
	\begin{equation}\label{eq:conddens}
	\begin{split}
	c_{j|(j+1):d}&(u_{j}|u_{j+1},\ldots,u_{d})=\\
	&\prod_{k=j+1}^{d}c_{m_{k,j},m_{j,j};m_{k+1,j},\ldots,m_{d,j}}\big(C_{m_{k,j}|m_{k+1,j},\ldots,m_{d,j}}\left(u_{m_{k,j}}|u_{m_{k+1,j}},\ldots,u_{m_{d,j}}\right),\\ 
	& \qquad \qquad  C_{m_{j,j}|m_{k+1,j},\ldots,m_{d,j}}\left(u_{m_{j,j}}|u_{m_{k+1,j}},\ldots,u_{m_{d,j}}\right)\!;u_{m_{k+1,j}},\ldots,u_{m_{d,j}}\big).
	\end{split}
	\end{equation}
\end{prop}
This proposition will prove itself to be crucial for the development of the distance measures from \autoref{sec:modeldist}. For simulation and Monte Carlo integration it is important that we can sample from vine copula distributions. \cite{stoeber2012} and \cite{joe2014dependence} provide sampling algorithms for arbitrary vine copulas. They are based on the inverse Rosenblatt transformation \citep{rosenblatt1952}: First, sample $w_j\sim \text{uniform}(0,1)$ for $j=1,\ldots,d$. Then, apply an inverse Rosenblatt transform $T_{c}$ to the uniform sample, i.e.\ $\ub=(u_1,\ldots,u_d)'=T_{c}(\wb)$, where $\wb=(w_1,\ldots,w_d)'$ is mapped from the (uniform) w-scale to the (warped) u-scale in the following way:
\begin{itemize}
	\item[$\bullet$] $u_d:=w_d$,
	\item[$\bullet$] $u_{d-1}:=C_{d-1| d}^{-1}(w_{d-1}| u_d)$,\\
	$\vdots$
	\item[$\bullet$] $u_1:=C_{1| 2:d}^{-1}(w_1| u_2,\ldots,u_d)$.
\end{itemize}
Note that the appearing inverse conditional distribution functions can be obtained easily for vine copulas.
When it comes to modeling, for tractability reasons most authors assume that for pair-copulas with a non-empty conditioning set the copula itself does not depend on the conditioning variables (e.g.\ $c_{13;2}(\,\cdot\,,\cdot\,;u_2)=c_{13;2}(\,\cdot\,,\cdot\,)$ for any $u_2\in[0,1]$). This assumption is referred to as the \emph{simplifying assumption}.\label{txt:sa} Among others, \cite{haff2010simplified}, \cite{acar2012beyond} and \cite{stoeber2013simplified} discuss when this assumption is justified. Since \textit{simplified vines}, i.e.\ vine copulas satisfying the simplifying assumption, are in practice the most relevant class of vine copulas for high dimensions, all examples in this paper consider simplified vines. Nevertheless, the presented concepts are also applicable to non-simplified vines (see \autoref{sec:conclusion}).\\

We typically work in a (simplified) parametric framework, where we specify each pair-copula of the vine decomposition as a parametric bivariate copula with up to two parameters. For the sake of notation, we borrow the concept of the vine structure matrix to introduce a lower-triangular family matrix $B=\left(b_{i,j}\right)_{i,j=1}^d$ and two lower-triangular parameter matrices $P^{(k)}=(p^{(k)}_{i,j})_{i,j=1}^d$, $k=1,2$, containing the pair-copula families and associated parameters of $c_{m_{i,j},m_{j,j}|m_{i+1,j},\ldots,m_{d,j}}$, respectively. Since we only use one- and two-parametric copula families, two parameter matrices are sufficient. The entries of the family and parameter matrices, $b_{i,j}$, $p^{(1)}_{i,j}$ and $p^{(2)}_{i,j}$, specify the pair-copula corresponding to the entry $m_{i,j}$. For one-parametric families we set the corresponding entry in the second parameter matrix to zero. For the family matrix, we use the following copula families (with corresponding abbreviations): Gaussian ($\Nc$), Student t (t), Clayton ($\Cc$), Gumbel ($\Gc$), Frank ($\Fc$) and Joe ($\Jc$). In order to compare the strengths of dependence of different copula families, we also compute the Kendall's $\tau$ values $k_{i,j}$ corresponding to pair-copulas with family $b_{i,j}$ and parameters $p^{(1)}_{i,j}$ and $p^{(2)}_{i,j}$ and store them in a lower-triangular matrix $K=\left(k_{i,j}\right)_{i,j=1}^d$. A (simplified) vine copula can then be written as the quadruple $\Rc=\left(M,B,P^{(1)},P^{(2)}\right)$.

\cite{dissmann2013selecting} developed a sequential estimation method that fits a simplified vine, i.e.\ the structure matrix as well the corresponding family and parameter matrices, to a given data set. This algorithm is also implemented in R \citep{R} as the function \texttt{RVineStructureSelect} in the package \texttt{VineCopula} \citep{VC}, which we use frequently throughout this paper.

Finally, for a simplified vine we define the associated \emph{nearest Gaussian vine}, i.e.\ the vine with the same structure matrix and Kendall's $\tau$ values but only Gaussian pair-copulas.
\begin{defi}[Nearest Gaussian vine]\label{defi:NGV}
	For a simplified vine copula $\Rc=(M,B,P^{(1)},P^{(2)})$ let $K=\left(k_{i,j}\right)_{i,j=1}^d$ denote the lower-triangular matrix containing the corresponding Kendall's $\tau$ values. Then, the \emph{nearest Gaussian vine} of $\Rc$ is given by $\tilde{\Rc}=(M,\tilde{B},\tilde{P}^{(1)},\tilde{P}^{(2)})$, where $\tilde{B}$ is a family matrix where all entries are Gaussian, $\tilde{P}^{(1)}=(\tilde{p}^{(1)}_{i,j})_{i,j=1}^d$ with $\tilde{p}^{(1)}_{i,j}=\sin\left(\frac{\pi}{2} k_{i,j}\right)$ and $\tilde{P}^{(2)}$ is a zero-matrix.
\end{defi}

\section{Model distances for vines}\label{sec:modeldist}

There are many motivations to measure the model distance between different vines. For example, \cite{stoeber2013simplified} try to find the simplified vine with the smallest distance to a given non-simplified vine. Further, it might be of interest to measure the distance between a vine copula and a Gaussian copula, both fitted to the same data set, in order to assess the need for the more complicated model. A common method to measure the distance between vines is the Kullback-Leibler distance.

\subsection{Kullback-Leibler distance}\label{subsec:KL}
\cite{kullback1951information} introduced a measure that indicates the distance between two $d$-dimensional statistical models with densities $f,g\colon\Rbb^d\to\left[0,\infty\right)$. The so-called \emph{Kullback-Leibler (KL) distance} between $f$ and $g$ is defined as
\begin{equation}\label{eq:KL}
\KL(f,g):=\int_{\xb\in\Rbb^d}\ln\left(\frac{f(\xb)}{g(\xb)}\right)f(\xb)\diff\xb.
\end{equation}
The KL distance between $f$ and $g$ can also be expressed as an expectation with respect to $f$:
\begin{equation}\label{eq:KL_expectation}
\KL(f,g)=\Ebb_f\left[\ln \left(\frac{f(\Xb)}{g(\Xb)}\right)\right],
\end{equation}
where $\Xb\sim f$. Note that the KL distance is non-negative and equal to zero if and only if $f=g$. It is not symmetric, i.e.\ in general $\KL(f,g)\neq\KL(g,f)$ for arbitrary densities $f$ and $g$.

Under the assumption that $f$ and $g$ have identical marginals, i.e.\ $f_j=g_j$, $j=1,\ldots,d$, we can show that the KL distance between $f$ and $g$ is equal to the KL distance between their corresponding copula densities. For this, let $c^f$ and $c^g$ be the copula densities corresponding to $f$ and $g$ and assume that $f$ and $g$, respectively, have the same marginal densities. Then, using Sklar's Theorem, we obtain
\be\label{eq:indepmargins}
\begin{split}
	\KL(f,g) & =\int_{\xb\in\Rbb^d}\ln\left(\frac{f(\xb)}{g(\xb)}\right)f(\xb)\diff\xb\\
	& =\int_{\xb\in\Rbb^d}\ln\left(\frac{c^f(F_1(x_1),\ldots,F_d(x_d))}{c^g(F_1(x_1),\ldots,F_d(x_d))}\right)c^f(F_1(x_1),\ldots,F_d(x_d))\prod_{j=1}^df_j(x_j)\diff x_1\cdots \diff x_d\\
	& =\int_{\ub\in[0,1]^d}\ln\left(\frac{c^f(u_1,\ldots,u_d)}{c^g(u_1,\ldots,u_d)}\right)c^f(u_1,\ldots,u_d)\diff u_1\cdots \diff u_d\\
	& =\KL\big(c^f,c^g\big),
\end{split}
\ee
where we applied the substitution $u_j:=F_j(x_j)$, $\diff u_j=f_j(x_j)\diff x_j$ for the third equality.

In this paper we are mainly interested in comparing different models that are obtained by fitting a data set. Since we usually first estimate the margins and afterwards the dependence structure, the assumption of identical margins is always fulfilled. Hence, we will in the following concentrate on calculating the Kullback-Leibler distance between copula densities.

Having a closer look at the definition of the KL distance, we see that for its calculation a $d$-dimensional integral has to be evaluated. In general, this cannot be done analytically and, further, is numerically infeasible in high dimensions. For example, \cite{schepsmeier2015efficient} stresses the difficulty of numerical integration in dimensions 8 and higher. In this section, we propose modifications of the Kullback-Leibler distance designed to be computationally tractable and still measure model distances adequately. These modifications are all based on the following proposition that shows that the KL distance between $d$-dimensional copula densities $c^{f}$ and $c^{g}$ can be expressed as the sum over expectations of KL distances between univariate conditional densities (see \autoref{App:KLdecomp} for a proof).
\begin{prop}\label{prop:KLdecomp}
	For two copula densities $c^f$ and $c^g$ it holds:
	\begin{equation}\label{eq:KLcKL}
	\KL\big(c^{f},c^{g}\big)=\sum_{j=1}^d\Ebb_{c^{f}_{(j+1):d}}\Big[\KL\Big(c^{f}_{j|(j+1):d}\left(\,\cdot\,|\Ub_{(j+1):d}\right),c^{g}_{j|(j+1):d}\left(\,\cdot\,|\Ub_{(j+1):d}\right)\Big)\Big],
	\end{equation}
	where $\Ub_{(j+1):d}\sim c^{f}_{(j+1):d}$ and ${(d+1)\!:\!d}:=\emptyset$.
\end{prop}

This proposition is especially useful if $c^{f}$ and $c^{g}$ are vine copula densities, since the appearing conditional densities can easily be obtained (see \autoref{prop:column}). As an example, for four-dimensional copula densities $c^{f}$ and $c^{g}$, we can write:
\begin{equation}\label{eq:4dKL}
\begin{split}
\KL\big(c^{f},c^{g}\big)&=\Ebb_{c^{f}_{2:4}}\left[\KL\left(c^{f}_{1|2:4}(\,\cdot\,|\Ub_{2:4}),c^{g}_{1|2:4}(\,\cdot\,|\Ub_{2:4})\right)\right]\\
&\,\,+\Ebb_{c^{f}_{3,4}}\left[\KL\left(c^{f}_{2|3,4}(\,\cdot\,|\Ub_{3,4}),c^{g}_{2|3,4}(\,\cdot\,|\Ub_{3,4})\right)\right]\\
&\,\,+\Ebb_{c^{f}_{4}}\left[\KL\left(c^{f}_{3|4}(\,\cdot\,|U_{4}),c^{g}_{3|4}(\,\cdot\,|U_{4})\right)\right]\\
&\,\,+0,
\end{split}
\end{equation}
where for instance
\begin{multline*}
c^f_{1|2:4}(u_1|u_2,u_3,u_4)=c^f_{12}(u_1,u_2)\,c^f_{13;2}\left(C^f_{1|2}(u_1|u_2),C^f_{3|2}(u_3|u_2);u_2\right) \\  \times c^f_{14;23}\left(C^f_{1|23}(u_1|u_2,u_3),C^f_{4|23}(u_4|u_2,u_3);u_2,u_3\right).
\end{multline*}
The zero in the last line of \autoref{eq:4dKL} results from the fact that $c^{f}_{4}(u_4)= c^{g}_{4}(u_4)= 1$ for all $u_4\in[0,1]$. This is generally the case for the $d$th summand in \autoref{eq:KLcKL}, which will therefore be omitted in the following.

Note that the evaluation of the KL distance with this formula still implicitly requires the calculation of a $d$-dimensional integral since the expectation in the first summand of \autoref{eq:KLcKL} demands a $(d-1)$-dimensional integral of the KL distance between univariate densities. A commonly used meth-od to approximate expectations is \textit{Monte Carlo (MC) integration} (see for example \cite{caflisch1998monte}): For a random vector $\Xb\in\Rbb^d$ with density $f\colon\Rbb^d\rightarrow[0,\infty)$ and a scalar-valued function $h\colon\Rbb^d\rightarrow\Rbb$, the expectation $\Ebb_f[h(\Xb)]=\int_{\Rbb^d} h(\xb)f(\xb)\diff \xb$ can be approximated by
\begin{equation}\label{eq:MCI}
\Ebb_f[h(\Xb)]\approx \frac{1}{N_{\text{MC}}}\sum_{i=1}^{N_{\text{MC}}}{h(\xb_i)},
\end{equation}
where $\left\lbrace \xb_i\right\rbrace _{i=1}^{N_{\text{MC}}}$ is an i.i.d.\ sample of size $N_{\text{MC}}$ distributed according to the density $f$. However, the slow convergence rate of this method has been subject to criticism. Moreover, \cite{do2003fast} argues that when approximating the KL distance via Monte Carlo integration the random nature of the method is an unwanted property. Additionally, MC integration might produce negative approximations of KL distances even though it can be shown theoretically that the KL distance is non-negative.

As an alternative to Monte Carlo integration, in the next sections we propose several ways to approximate the expectation in \autoref{eq:KLcKL} by replacing it with the average over a $(d-j)$-dimensional non-random grid $\Uc_j$, such that
\begin{equation}\label{eq:KLapprox}
\KL\big(c^f,c^g\big) \approx\sum_{j=1}^{d-1}\frac{1}{\left|\,\Uc_j\right| }\sum_{\ub_{(j+1):d}\in\, \Uc_j} \!\!\!\KL\Big(c^{f}_{j|(j+1):d}(\,\cdot\,|\ub_{(j+1):d}),c^{g}_{j|(j+1):d}(\,\cdot\,|\ub_{(j+1):d})\Big).
\end{equation}
Note that, being a sum over univariate KL distances, this approximation produces non-negative results, regardless of the grids $\Uc_j$, $j=1,\ldots,d$. Now, the question remains how to choose the grids $\Uc_j$, such that the approximation is on the one hand fast to calculate and on the other hand still maintains the main properties of the KL distance. We will provide three possible answers to this question yielding different distance measures and investigate their performances.

Throughout the subsequent sections we assume the following setting: Let $\Rc^f$ and $\Rc^g$ be two $d$-dimensional vines with copula densities $c^f$ and $c^g$, respectively. We assume that their vine structure matrices have the same entries on the diagonals, i.e.\ $\diag(M^f)=\diag(M^g)$. Note that, although this assumption is a restriction, there are still $2^{\binom{d-2}{2}+d-2}$ different vine decompositions with equal diagonals of the structure matrix (cf.\ \autoref{prop:numbervines}).\footnote{This includes, for example, C- and D-vines \citep{aasczado} having the same diagonal.} As before, without loss of generality we set the diagonals equal to ${1\!:\!d}$.

\subsection{Approximate Kullback-Leibler distance}\label{subsec:aKL}

We illustrate the idea of the approximate Kullback-Leibler distance at the example of two three-dimensional vines $\Rc^f$ and $\Rc^g$. For the first summand ($j=1$) of \autoref{eq:KLapprox}, the KL distance between $c^f_{1|2,3}(\,\cdot\,|u_2,u_3)$ and $c^g_{1|2,3}(\,\cdot\,|u_2,u_3)$ is calculated for all pairs $(u_2,u_3)'
$ contained in the grid $\Uc_1$. In this example we assume that the pair-copula $c^f_{2,3}$ is a Gumbel copula with parameter $\theta=6$ (implying a Kendall's $\tau$ value of $0.83$). Regarding the choice of the grid, if we used the Monte Carlo method, $\Uc_1$ would contain a random sample of $c^f_{2,3}$. Recall from \autoref{sec:vines} that such a sample can be generated by simulating from a uniform distribution on $[0,1]^2$ and applying the inverse Rosenblatt transformation $T_{c^f_{2,3}}$. \autoref{fig:WU-Plots_MC} displays a sample of size 900 on the (uniform) w-scale and its transformation via $T_{c^f_{2,3}}$ to the (warped) u-scale.
\begin{figure}[!htb]
	\centering
	\includegraphics[width=0.4\linewidth]{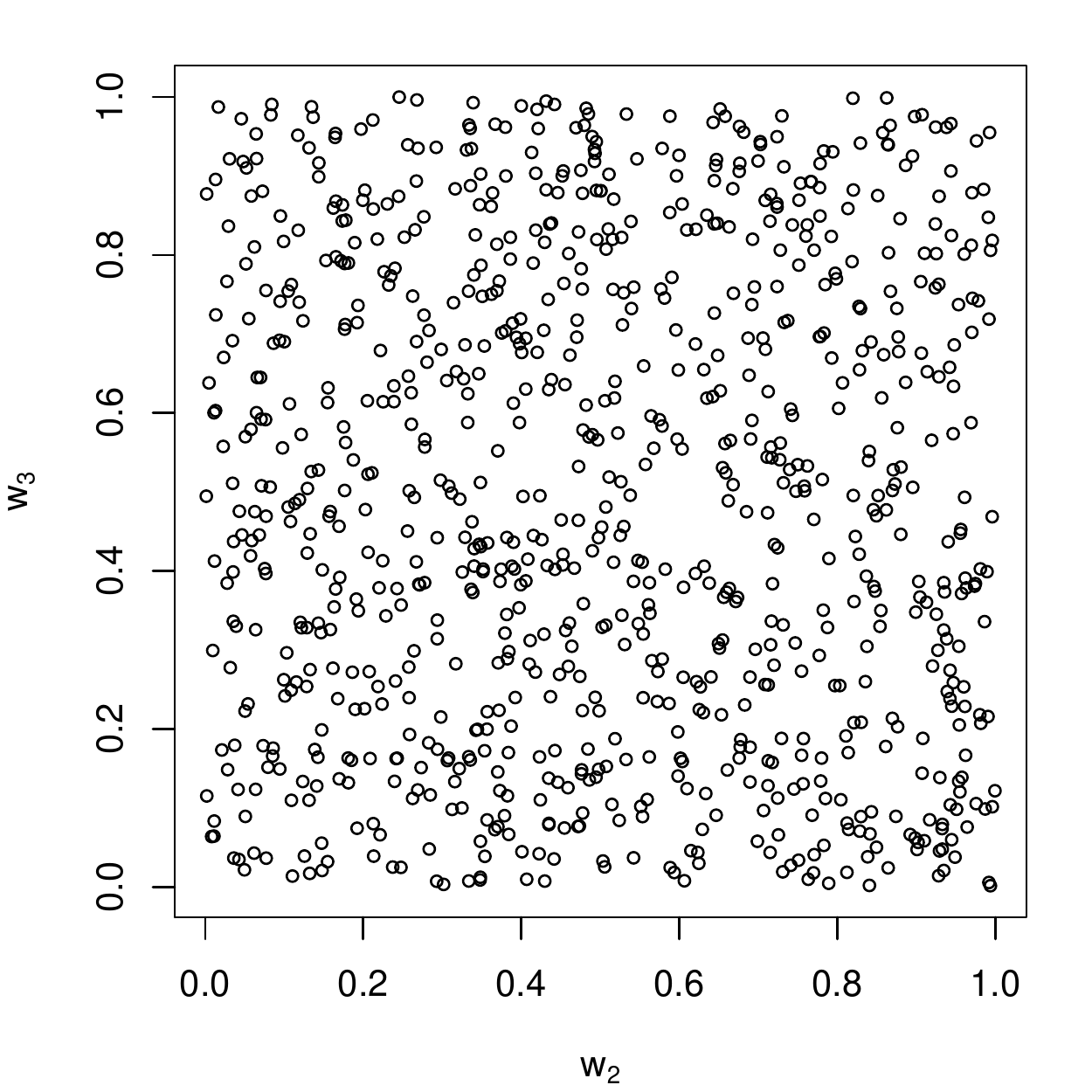}
	\includegraphics[width=0.4\linewidth]{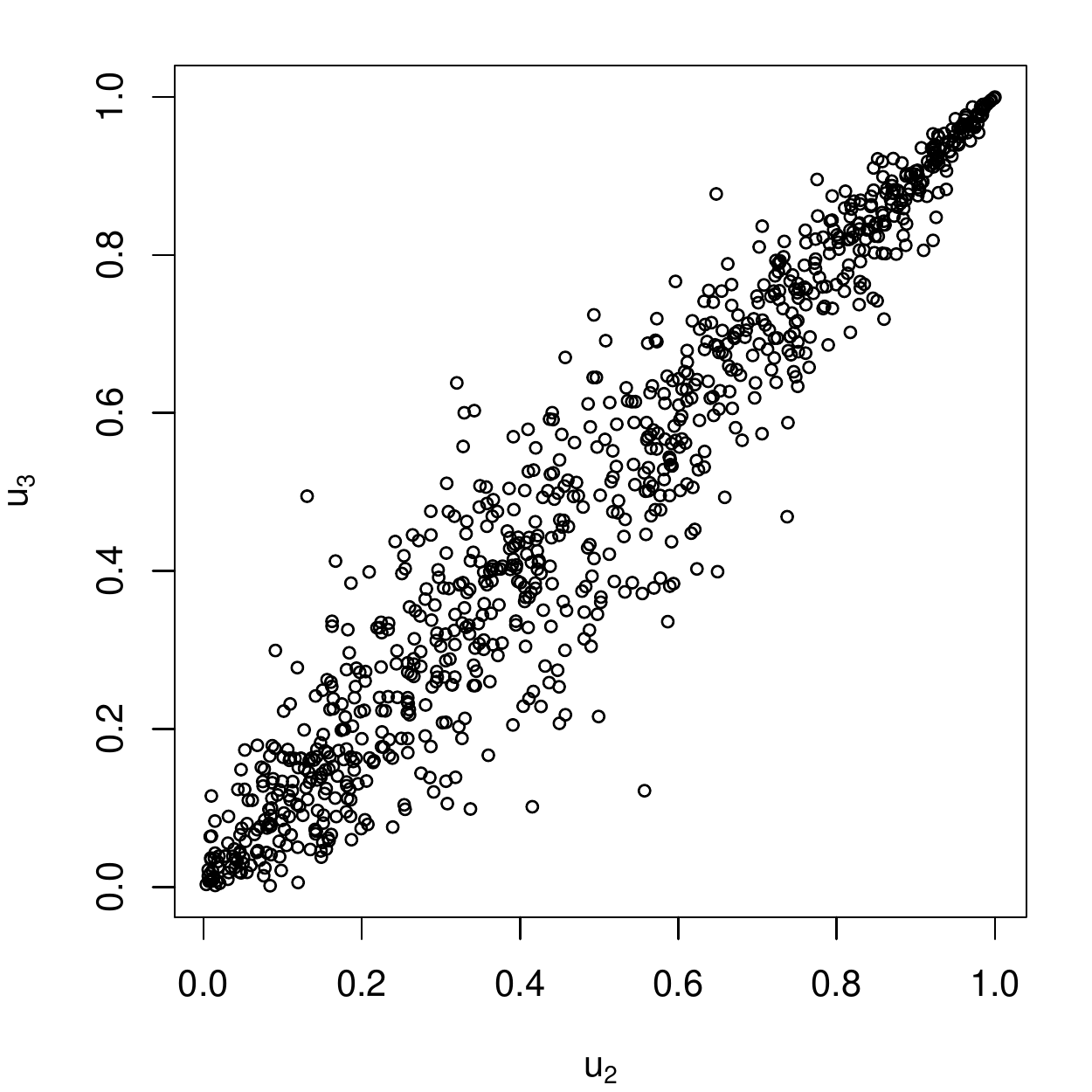}
	\caption{Sample of size 900 from the uniform distribution (left) and corresponding warped sample under transformation $T_{c^f_{2,3}}$, which is a sample from a Gumbel copula with $\theta=6$ (right).}
	\label{fig:WU-Plots_MC}
\end{figure}\\

As mentioned before we do not want our distance measure to be random. This motivates us to introduce the concept of \emph{structured Monte Carlo integration}: Instead of sampling from the uniform distribution on the w-scale, we use a \emph{structured grid} $\Wc$, which is an equidistant lattice on the two-dimensional unit cube\footnote{Since most copulas have an infinite value at the boundary of the unit cube, we usually restrict ourselves to $[\eps,1-\eps]^d$ for a small $\eps>0$.}, and transform it to the warped u-scale by applying the inverse Rosenblatt transformation $T_{c^f_{2,3}}$. \autoref{fig:WU-Plots_aKL} shows an exemplary structured grid with 30 grid points per margin.
\begin{figure}[!htb]
	\centering
	\includegraphics[width=0.4\linewidth]{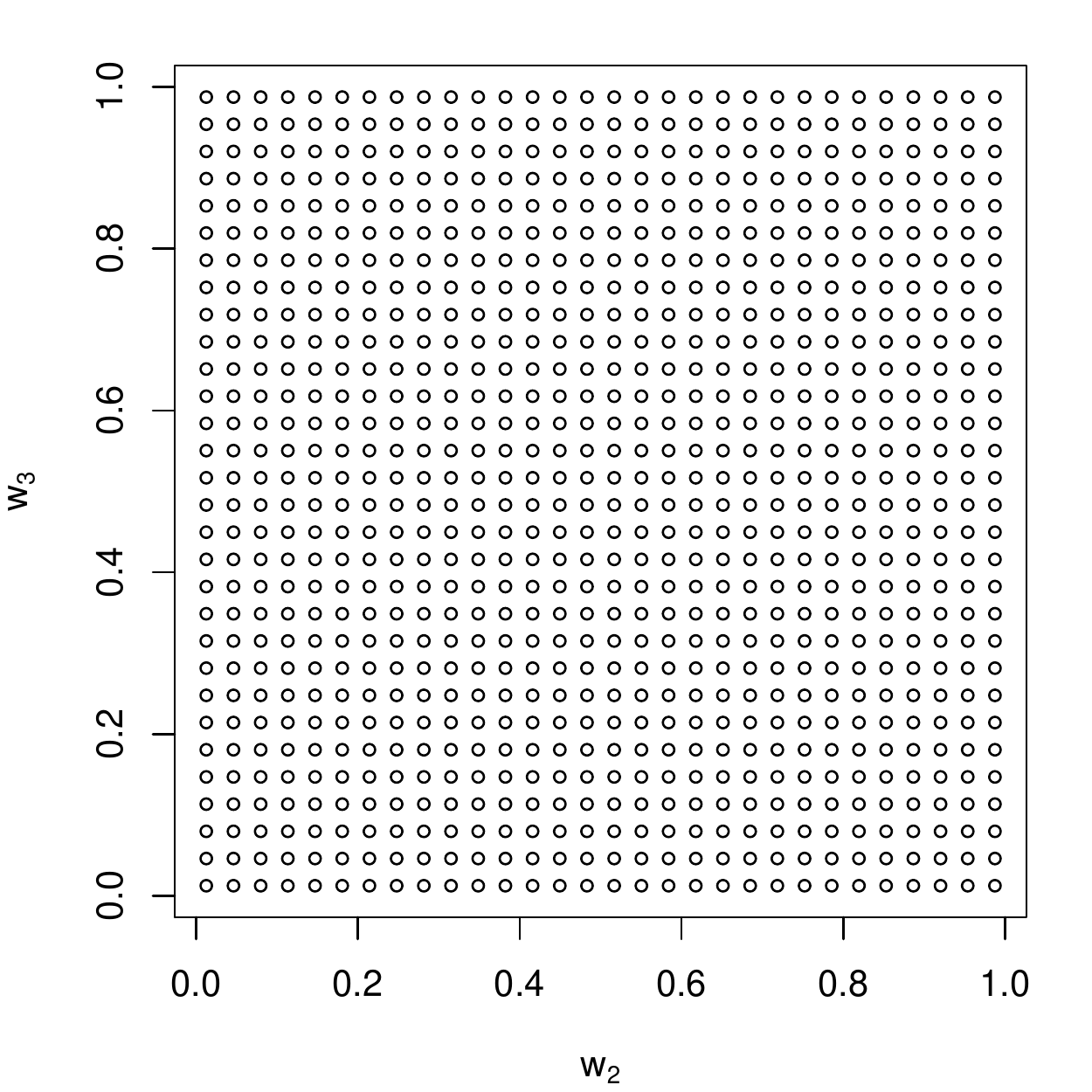}
	\includegraphics[width=0.4\linewidth]{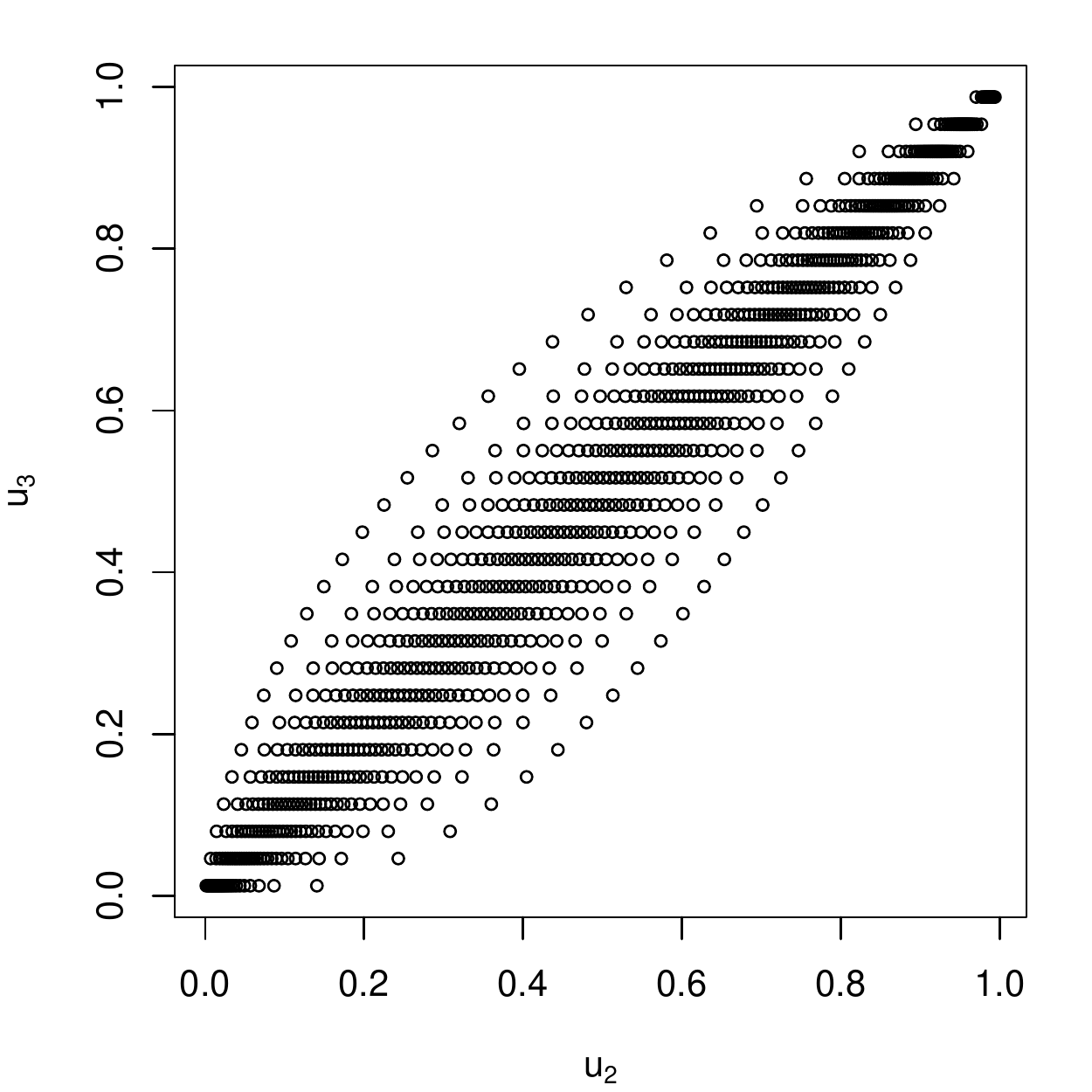}
	\caption{Structured grid with 30 grid points per margin (left) and corresponding warped grid under transformation $T_{c^f_{2,3}}$ (right).}
	\label{fig:WU-Plots_aKL}
\end{figure}\\
Applying this procedure for all grids $\,\Uc_j$, $j=1,\ldots,d-1$, yields the \emph{approximate Kullback-Leibler distance}.

\begin{defi}[Approximate Kullback-Leibler distance]\label{defi:aKL}
	Let $\Rc^f$ and $\Rc^g$ be as described above. Further, let $n\in\Nbb$ be the number of grid points per margin and $\eps>0$. Then, the \emph{approximate Kullback-Leibler distance} ($\aKL$) between $\Rc^{f}$ and $\Rc^{g}$ is defined as
	\begin{equation*}
	\aKL\big(\Rc^{f},\Rc^{g}\big):=\sum_{j=1}^{d-1} \frac{1}{\left|\Gc_j \right| }\sum_{\ub_{(j+1):d}\in\Gc_{j}} \!\!\!\KL\Big(c^{f}_{j|(j+1):d}\left(\,\cdot\,|\ub_{(j+1):d}\right)\!,c^{g}_{j|(j+1):d}\left(\,\cdot\,|\ub_{(j+1):d}\right)\Big),
	\end{equation*}
	where the warped grid $\Gc_{j}\subseteq [0,1]^{d-j}$ is constructed as follows:
	\begin{enumerate}
		\item Define the structured grid $\Wc_{j}:=\left\{\eps,\eps+\Delta,\ldots,1-\eps\right\}^{d-j}$ to be an equidistant discretization of $[0,1]^{d-j}$ with $n$ grid points per margin, where $\Delta:=\frac{1-2\eps}{n-1}$.
		\item The warped grid $\Gc_{j}:=T_{c^{f}_{(j+1):d}}(\Wc_j)$ is defined as the image of $\Wc_j$ under the inverse Rosenblatt transform $T_{c^{f}_{(j+1):d}}$ associated with the copula density $c^{f}_{(j+1):d}$.
	\end{enumerate}
	Note that by construction $\left|\Gc_j \right|=n^{d-j}.$
\end{defi}
\autoref{prop:aKL_limit} shows that the approximate KL distance in fact approximates the true KL distance in the sense that the aKL converges to the KL for $\eps\to0$ and $n\to\infty$. A proof can be found in \autoref{App:aKL_limit}.
\begin{prop}\label{prop:aKL_limit}
	Let $\Rc^f$ and $\Rc^g$ be as described above. Then,
	\[
	\lim_{\eps\to 0}\lim_{n\to\infty}\aKL\big(\Rc^f,\Rc^g\big) = \KL\big(c^f,c^g\big).
	\]
	
\end{prop}
In the following applications we use the function \texttt{integrate} for the calculation of the one-dimensional KL. Further, we choose $\eps$ such that the convex hull of the structured grid contains volume $\beta\in(0,1)$, so $\eps:=\frac{1}{2}\big(1-\beta^{\frac{1}{d-j}}\big)$. 
Unless otherwise specified we set $\beta$ to be 95\%.

\begin{exa}[Four-dimensional aKL-example]\label{exa:4d}
	We consider a data set from the Euro Stoxx 50, already used in \cite{brechmann2013risk}. It covers a 4-year period (May 22, 2006 to April 29, 2010) containing 985 daily observations. The Euro Stoxx 50 is a major index consisting of the stocks of 50 large European companies. In this example we consider the following four national indices: the Dutch AEX ($U_1$), the Italian FTSE MIB ($U_2$), the German DAX ($U_3$) and the Spanish IBEX 35 ($U_4$). Fitting a simplified vine to the data yields:
	\[M=\begin{pmatrix}
	1 & 0 & 0 & 0 \\
	4 & 2 & 0 & 0 \\
	2 & 4 & 3 & 0 \\
	3 & 3 & 4 & 4 \\
	\end{pmatrix},\,\,\,
	B=\begin{pmatrix}
	0 & 0 & 0 & 0 \\
	\Fc & 0 & 0 & 0 \\
	\text t &\text  t & 0 & 0 \\
	\text t &\text  t &\text  t & 0 \\
	\end{pmatrix},\]
	\[P^{(1)}=\begin{pmatrix}
	0 & 0 & 0 & 0 \\
	1.01 & 0 & 0 & 0 \\
	0.36 & 0.36 & 0 & 0 \\
	0.91 & 0.89 & 0.88 & 0 \\
	\end{pmatrix}, \,\,\, P^{(2)}=\begin{pmatrix}
	0 & 0 & 0 & 0 \\
	0 & 0 & 0 & 0 \\
	6.34 & 10.77 & 0 & 0 \\
	6.23 & 4.96 & 6.80 & 0 \\
	\end{pmatrix}.\]
	As usual for financial data, most of the pair-copulas selected by the fitting algorithm are t copulas with rather high dependence; only $c_{14;23}$ is modeled as a Frank copula. Now we compute the approximate KL distance between this vine and its nearest Gaussian vine (see \autoref{defi:NGV}) and compare it to the numerically integrated KL distance. The latter limits our example to low dimensions because numerical integration becomes very slow in higher dimensions. Even in four dimensions we have to set the tolerance level of the integration routine \texttt{adaptIntegrate} of the package \texttt{cubature} from its default level of $10^{-5}$ to $10^{-4}$ to obtain results within less than 10 days. Throughout the paper we will also consider examples for $d\geq5$, where numerical integration becomes (almost) infeasible. As a substitute benchmark for the numerically integrated KL distance, we compare our approximated KL values to the corresponding Monte Carlo Kullback Leibler (MCKL) values, where the expectation in \autoref{eq:KL_expectation} is approximated by Monte Carlo integration, i.e.\
	\begin{equation}
	\MCKL\big(c^{f},c^{g}\big):=\frac{1}{N_{\text{MC}}}\sum_{i=1}^{N_{\text{MC}}}\ln \left(\frac{c^{f}(\ub^i)}{c^{g}(\ub^i)}\right),
	\end{equation}
	where $\ub^1,\ldots,\ub^{N_{\text{MC}}}$ are sampled from $c^{f}$. We choose the sample size $N_{\text{MC}}$ to be very large in order to get acceptable low-variance results \citep[cf.][]{do2003fast}.
	
	\autoref{tab:aKL} displays the approximate Kullback-Leibler distance between the fitted vine and its nearest Gaussian vine for different values of $\beta$ and $n$ together with the corresponding computational time (in hours). We further present the numerically and Monte Carlo integrated KL distances. In order to facilitate comparability, for each value of $\beta$ we compute the integrals on the corresponding domain of integration with volume $\beta$.
	
	\begin{table*}[!htb]
		\centering
		\begin{tabular}{|c|l|ccc|c|cc|}
			\cline{3-8}
			\multicolumn{2}{c|}{}& \multicolumn{3}{c|}{aKL} & Numeric &\multicolumn{2}{c|}{MCKL}\\
			\hline
			$\beta$ & & $n=10$ & $n=20$ & $n=50$  & tol=$10^{-4}$ & $N_{\text{MC}}=10^5$ & $N_{\text{MC}}=10^6$ \\
			\hline
			\multirow{2}{*}{95\%} & value&$0.135$ & $0.095$ & $0.076$ &   0.077 & 0.076 & 0.079 \\
			& time [h] & 0.004 & 0.030 & 0.582 & 20.3 & 0.005 & 0.061 \\
			\hline
			\multirow{2}{*}{99\%} & value & 0.311 & 0.170 & 0.107 & 0.082 & 0.085 & 0.081  \\
			& time [h] & 0.006 & 0.034 & 0.609 & 33.4 & 0.006 & 0.063\\
			\hline
			\multirow{2}{*}{100\%} & value &&& & 0.084 & 0.084 & 0.084 \\
			& time [h] &&& &  99.4 & 0.005 & 0.058 \\
			\hline
		\end{tabular}
		\caption{Approximate, numerically integrated and Monte Carlo integrated KL distances for different parameter settings with corresponding computational times (in hours).}
		\label{tab:aKL}
	\end{table*}

	We see that for an increasing number of marginal grid points $n$ the value of the approximate KL distance gets closer to the value obtained by numerical integration. We further observe that in this example the value of the numerically integrated KL distance does not change considerably when the integral is computed on the constrained domain of integration with volume $\beta$. We expect the computational time of the aKL to increase cubically since the number of univariate KL evaluations is $\sum_{j=1}^3|\Gc_{j}|=n^3+n^2+n$. This is empirically validated by the observed computational times. 
	Further, we see that even for larger values of $n$ the $\aKL$ is still considerably faster than classical numerical integration. Concerning the Monte Carlo integrated KL distances in this example, we observe that the values still vary notably between $N_{\text{MC}}=10^5$ and $N_{\text{MC}}=10^6$. Thus, for the remainder of the paper we will use $N_{\text{MC}}=10^6$ in order to get rather reliable results.
\end{exa}
We can conclude that the approximate KL distance is a valid tool to estimate the Kullback-Leibler distance. However, similar to numerical integration it suffers from the curse of dimensionality, causing computational times to increase sharply when a certain precision is required or dimension increases. The number of evaluation points $\left|\Gc_j \right|$ increases exponentially in $d$, making calculations infeasible for higher dimensions. This motivates us to thin out the grids $\Gc_{j}$ in a way that considerably reduces the number of grid points, while still producing sound results. We have found that the restriction to diagonals in the unit cube fulfills these requirements reasonably well. Of course, with this modification we cannot hope for the resulting distance measure to still approximate the KL distance but we will see that in applications it reproduces the behavior of the original KL distance remarkably well.

\subsection{Diagonal Kullback-Leibler distance}\label{subsec:dKL}

In order to illustrate the idea behind the \textit{diagonal Kullback-Leibler distance} we continue our example from \autoref{subsec:aKL}. \autoref{fig:WU-Plots_dKL} shows the structured and warped grids used for the aKL (gray circles). Additionally, the diagonal grid points are highlighted by filled diamonds.

\begin{figure}[!htb]
	\centering
	\includegraphics[width=0.4\linewidth]{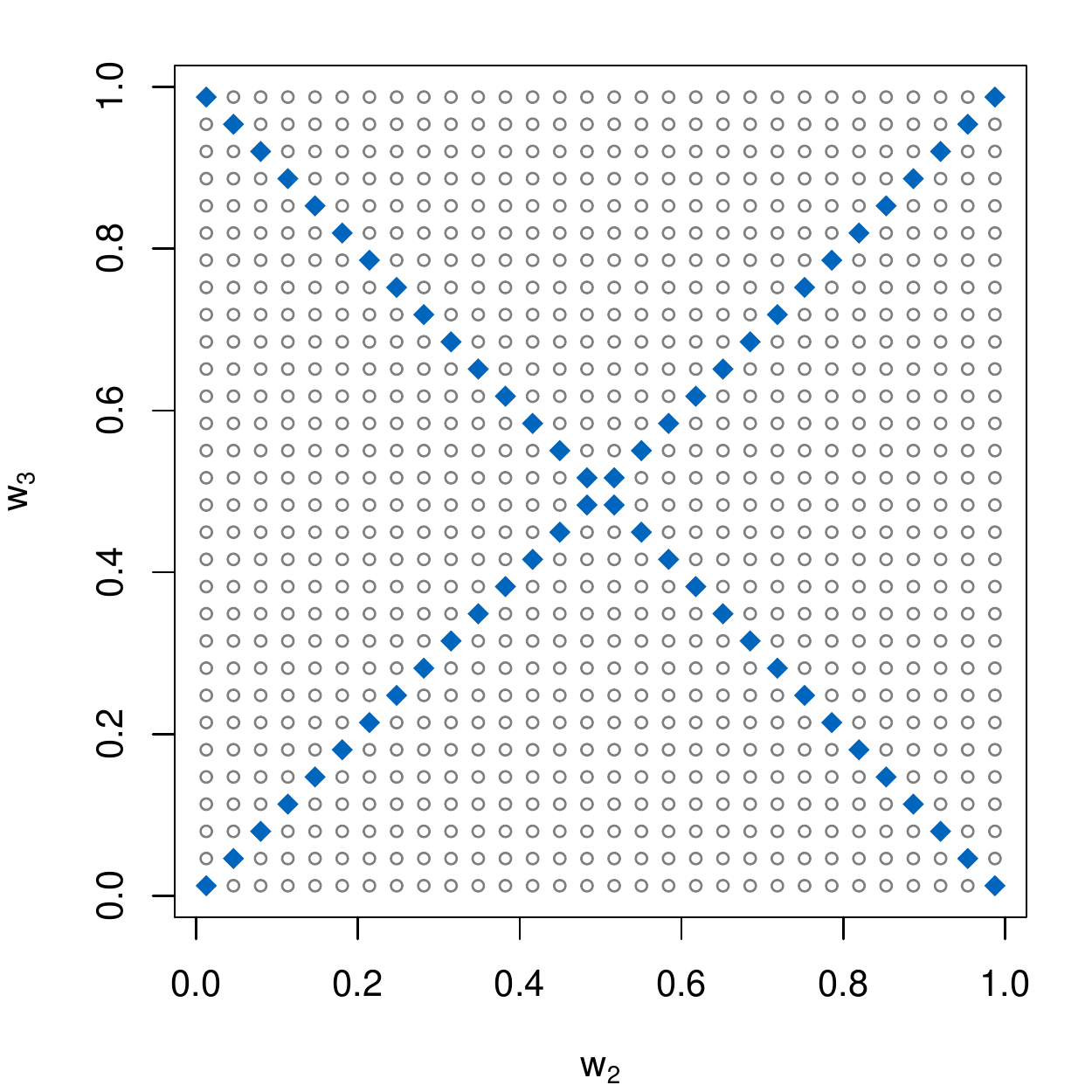}
	\includegraphics[width=0.4\linewidth]{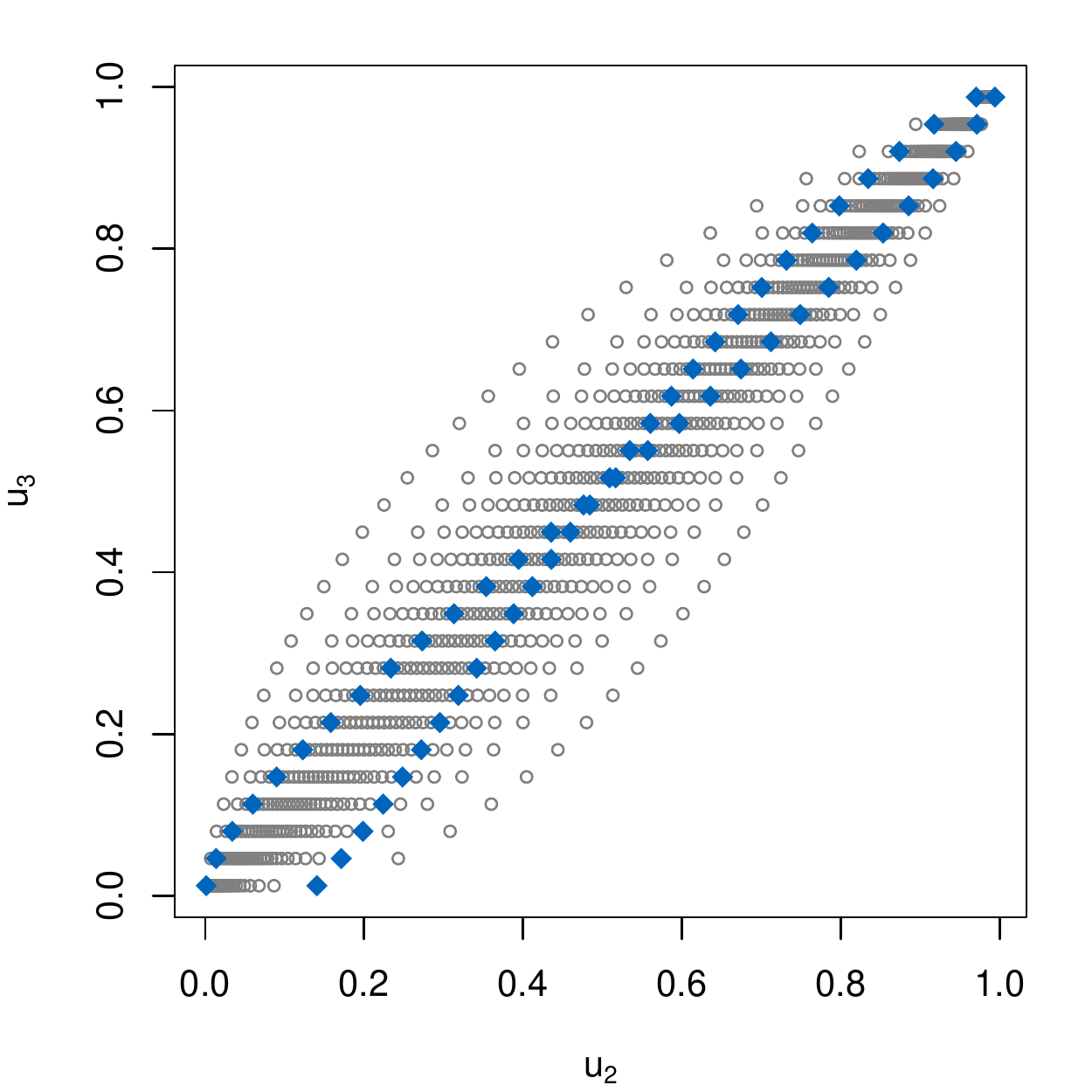}
	\caption{Structured grid with highlighted diagonals consisting of 30 evaluation points (left) and corresponding warped grid under transformation $T_{c^f_{2,3}}$ (right).}
	\label{fig:WU-Plots_dKL}
\end{figure}
The idea is now to reduce the evaluation grids $\Uc_j$ to the diagonal grids in order to define a distance measure related to the original KL distance with the advantage of reduced computational costs. For this, we formally define the sets of diagonals and warped discretized diagonals.

\begin{defi}[Diagonals and warped discretized diagonals]\label{defi:diagonals}
	For $j=1,\ldots,d-1$, we define the \emph{set of diagonals} in $[0,1]^{d-j}$:
	\[
	D_j:=\left\lbrace \left\lbrace \rb+t \vb(\rb)\mid t\in[0,1] \right\rbrace\, \middle|\,\rb\in\left\lbrace 0,1 \right\rbrace^{d-j}  \right\rbrace,
	\]
	where $\vb(\rb)=\left(v_1(\rb),\ldots,v_{d-j}(\rb)\right)'$ with
	\[ v_i(\rb):=
	\begin{cases}
	1&\text{if }\,r_i=0\\
	-1&\text{if }\,r_i=1
	\end{cases}, \quad i=1,\ldots,d-j,\] is the direction vector corresponding to the corner point $\rb$. Note that the set of diagonals $D_j$ only contains $2^{d-j-1}$ elements since every diagonal is implied by two corner points (e.g.\ the diagonals $\{(0,\ldots,0)'+t(1,\ldots,1)'\mid t\in[0,1] \}$ and
	$\{(1,\ldots,1)'+t(-1,\ldots,-1)'\mid t\in[0,1] \}$ coincide). Further, let $D_{j,1},\ldots,D_{j,2^{d-j-1}}$ be an arbitrary ordering of the $2^{d-j-1}$ diagonals. We define the $k$th \emph{discretized diagonal} on the w-scale as $\Dc^w_{j,k}:=D_{j,k} \cap \Wc_{j}$, where $\Wc_{j}$ is the structured grid in $[0,1]^{d-j}$ defined in \autoref{defi:aKL}, such that it contains $n$ grid points (cf.\ left panel of \autoref{fig:WU-Plots_dKL}). Finally, the $k$th \emph{warped discretized diagonal} on the u-scale is defined as $\Dc^u_{j,k}:=T_{c^{f}_{(j+1):d}}\big(\Dc^w_{j,k}\big)$, where $T_{c^{f}_{(j+1):d}}$ is defined as in \autoref{defi:aKL} (cf.\ right panel of \autoref{fig:WU-Plots_dKL}).
\end{defi}

Now, we can define the diagonal Kullback-Leibler distance by using the set of warped discretized diagonals
\[
\Dc_j^u:=\bigcup_{k=1}^{2^{d-j-1}}\Dc_{j,k}^u\]
as evaluation grid $\Uc_j$.

\begin{defi}[Diagonal Kullback-Leibler distance]\label{defi:dKL}
	Let $\Rc^f$, $\Rc^g$ and $\Dc_j^u$ be as described above. Then, the \emph{diagonal Kullback-Leibler distance} $\dKL$ between $\Rc^{f}$ and $\Rc^{g}$ is defined as
	\begin{equation*}
	\dKL\big(\Rc^{f},\Rc^{g}\big):=\sum_{j=1}^{d-1} \frac{1}{\big| \Dc_j^u\big|} \sum_{\ub_{(j+1):d}\in \Dc^u_{j}} \\ \KL\Big(c^{f}_{j|(j+1):d}\left(\,\cdot\,|\ub_{(j+1):d}\right),c^{g}_{j|(j+1):d}\left(\,\cdot\,|\ub_{(j+1):d}\right)\Big),
	\end{equation*}
	where $\big| \Dc_j^u\big|=n\cdot 2^{d-j-1}$.
\end{defi}

\begin{rem}\label{rem:dKL}
	Similar to \autoref{prop:aKL_limit} one can show (see \aref{App:dKL_approx}) that for each of the $2^{d-j-1}$ diagonals $\Dc_{j,k}^u$ it holds
	\begin{multline*}
	\lim_{\eps\to 0}\lim_{n\to\infty} \frac{1}{n}\sum_{\ub_{(j+1):d}\in \Dc^u_{j,k}}\KL\left(c^{f}_{j|(j+1):d}\left(\,\cdot\,|\ub_{(j+1):d}\right),c^{g}_{j|(j+1):d}\left(\,\cdot\,|\ub_{(j+1):d}\right)\right) =\\ \frac{1}{\sqrt{d-j}} \! \int\limits_{\ub_{(j+1):d}\in D_{j,k}^u}\!\!{\KL\left(c^{f}_{j|(j+1):d}\left(\,\cdot\,|\ub_{(j+1):d}\right),c^{g}_{j|(j+1):d}\left(\,\cdot\,|\ub_{(j+1):d}\right)\right)c_{(j+1):d}^f\left(\ub_{(j+1):d}\right)\diff \ub_{(j+1):d}},
	\end{multline*}
	where $D_{j,k}^u:=T_{c_{(j+1):d}^f}(D_{j,k})$. Hence, the diagonal Kullback-Leibler distance can be interpreted as a sum of scaled approximated line integrals over weighted univariate KL distances between conditional densities, which is exactly the integrand appearing in \autoref{prop:KLdecomp}. Having the infeasibility of the aKL in higher dimensions in mind, the reduction to points on lines seems to be a good choice in order to reduce the approximation of multivariate integrals to one-dimensional ones. We choose the warped diagonals as these lines since they on the one hand contain points with high density values due to the warping and on the other hand each let all components of the conditioning vector take values on the whole range from 0 to 1. Since in practice vine models tend to differ most in the tails of the distributions, we increase the concentration of evaluation points in the tails by transforming the discretized diagonal with the method described at the end of \aref{App:tail_trafo}.
\end{rem}
The following examples are supposed to illustrate that the dKL is a reasonable distance measure between vine copulas.

\subsubsection*{Applications of the dKL}

\begin{exa}[\autoref{exa:4d} continued]\label{exa:4d_cont}
	We continue \autoref{exa:4d} and apply the dKL to measure the distance between the fitted four-dimensional vine and its nearest Gaussian vine using different numbers of grid points $n$ per diagonal and $\beta=95\%$ as usual. The results and computational times (in seconds) are displayed in \autoref{tab:ESdKL}.
	\begin{table}[!htbp]
		\centering
		\begin{tabular}{|l|cccccc|}
			\hline
			$n$ & $10$ & $20$ & $50$ & $100$ & $1{,}000$ & $10{,}000$ \\
			\hline
			value & 0.123 & 0.117 & 0.115 & 0.114 & 0.113 & 0.113\\
			time [s] & 0.9 & 1.7 & 3.9 & 7.7 & 89 & 964\\
			\hline
		\end{tabular}
		\caption{dKL values and corresponding computational times (in seconds) for the four-dimensional Euro Stoxx 50 example.}
		\label{tab:ESdKL}
	\end{table}
	
	We observe that the dKL values seem to converge quite fast and are already quite close to their limit for small $n$. Of course, we cannot expect them to converge to the original KL distance, but we see that the order of magnitude is the same as the values in \autoref{exa:4d} that were calculated by numerical integration. Concerning computational times, the dKL merely needs a couple of seconds to produce reasonable results which is a vast improvement to the computational times of the aKL and MCKL which were in the order of hours. As expected, the computational times of the dKL grow linearly in $n$.
\end{exa}

In order to assess the aptitude of the dKL for measuring the distance between vine copulas, we conduct several plausibility checks in the following examples. Since numerical integration is not practicable for these examples ($d\geq 5)$, we compare our dKL values to the corresponding MCKL values.

\begin{exa}[Plausibility checks]\label{exa:dKL_Plausi}
	For the first plausibility check, we consider five-dimensional t copulas with $\nu$ degrees of freedom (ranging from 3 to 30). Those can be specified as vine copulas with structure matrices $M=D_5$, family matrices $B$ containing only t copulas, Kendall's $\tau$ matrices $K=K_5(0.5)$ and degrees of freedom matrix $P^{(2)}=V_5(\nu)$, where
	\begin{equation}\label{eq:D_d}
	D_d:=\begin{pmatrix}
	1 &  & &  & & &\\
	d & 2 & & & & &\\
	d-1 & d & \ddots& & & & \\
	\vdots & d-1 & \ddots & \ddots&  & &\\
	4 & \vdots & \ddots &  \ddots& d-2& & \\
	3 & 4 & & \ddots& d& d -1& \\
	2 & 3 & 4 & \cdots& d-1 & d & d \\
	\end{pmatrix},
	\end{equation}
	\begin{equation}\label{eq:K_d}
	K_d(\tau)=\left(k_{i,j}(\tau)\right)_{i,j=1}^d,
	\end{equation}
	\[
	V_d(\nu)=\left(v_{i,j}(\nu)\right)_{i,j=1}^d.
	\]
	Here, $k_{i,j}(\tau):=\left(\frac{1}{2}\right)^{d-i}\cdot\tau\cdot 1_{\left\lbrace i>j\right\rbrace}$ and $v_{i,j}(\nu):=\left(\nu+d-i\right)\cdot 1_{\left\lbrace i>j\right\rbrace}$, where $1_{\left\lbrace \cdot \right\rbrace }$ denotes the indicator function. Note that the structure encoded in $D_d$ is also known as a D-vine structure \citep[see][]{aasczado} and the parameter matrix $P^{(1)}$ is uniquely determined by the Kendall's $\tau$ matrix since all pairs are bivariate t copulas. \autoref{tab:dKL_Plausi1} (left table) displays the diagonal ($n=10$) and Monte Carlo ($N_{\text{MC}}=10^6$) KL distances between these t copulas and their nearest Gaussian vines.
	\begin{table*}[!htbp]
		\centering
		\begin{tabular}{|r|r|r|r|}
			\hline
			$\nu$ & dKL & MCKL & ratio \\
			\hline
			3 & 0.857 & 0.374 & 2.29 \\
			5 & 0.376 & 0.162 & 2.32 \\
			7 & 0.209 & 0.091 & 2.30 \\
			10 & 0.109 & 0.047 & 2.31 \\
			15 & 0.051 & 0.023 & 2.26 \\
			20 & 0.029 & 0.013 & 2.20 \\
			25 & 0.019 & 0.008 & 2.25 \\
			30 & 0.013 & 0.006 & 2.23 \\
			\hline
		\end{tabular}
		\,
		\begin{tabular}{|r|r|r|r|}
			\hline
			$\tau$ & dKL & MCKL & ratio \\
			\hline
			$-0.7$ & 4.702 & 3.226 & 1.46 \\
			$-0.5$ & 2.106 & 1.431 & 1.47 \\
			$-0.3$ & 0.740 & 0.473 & 1.56 \\
			$-0.1$ & 0.077 & 0.050 & 1.54 \\
			0.1 & 0.067 & 0.048 & 1.41 \\
			0.3 & 0.561 & 0.423 & 1.33 \\
			0.5 & 1.740 & 1.262 & 1.38 \\
			0.7 & 4.590 & 2.982 & 1.54 \\
			\hline
		\end{tabular}
		\,
		\begin{tabular}{|r|r|r|r|}
			\hline
			family & dKL & MCKL & ratio \\
			\hline
			$\Nc$ &  0.369 & 0.205 & 1.80 \\
			$\Cc$ & 2.987 & 1.780 & 1.68 \\
			$s\Cc$ & 0.322 & 0.158 & 2.04 \\
			$\Jc$ & 0.483 & 0.249 & 1.94 \\
			\hline
		\end{tabular}
		\caption{Left table: dKL ($n=10$) and MCKL ($N_{\text{MC}}=10^6$) values between five-dimensional t copulas with $K=K_5(0.5)$ and $P^{(2)}=V_5(\nu)$ and their nearest Gaussian vines. Middle table: dKL and MCKL values between five-dimensional t copulas with $K=K_5(\tau)$ and $P^{(2)}=V_5(3)$ and their nearest Gaussian vines. Right table: dKL and MCKL values between a five-dimensional Gumbel D-vine and D-vines with the same Kendall's $\tau$ matrix constructed using one copula family only. The third column contains the ratio between dKL and MCKL.}
		\label{tab:dKL_Plausi1}
	\end{table*}
	
	We see that the diagonal KL distance decreases as the degrees of freedom increase. This is very plausible since the t copula converges to the Gaussian copula as $\nu\rightarrow\infty$. Further, while their values are not on the same scale, we observe that the dKL and MCKL values behave similarly. This can be seen by the fact that in this example the ratio between both values ranges only between $2.20$ and $2.32$, where some of the fluctuation can be explained by the randomness of the MCKL. Further, the fact that the scale of the dKL differs from the one of the KL is no real drawback since the scale of the KL distance itself is not particularly meaningful. Regarding computational times we note that in this five-dimensional example the average time for computing a MCKL value was 125 seconds, while the average computation of the dKL only took 3 seconds.
	
	In the second plausibility check we also deal with five-dimensional t copulas decomposed as above. However, in this scenario the degrees of freedom are fixed to be equal to 3 and the value for $\tau$ in $K_5(\tau)$ is ranging between $-0.7$ and $0.7$. As a reference vine we use a t copula with Kendall's $\tau$ matrix $K_5(0)$ and degrees of freedom $\nu=3$. The dKL and MCKL distances between the resulting eight t copulas and the reference vine is shown in \autoref{tab:dKL_Plausi1} (middle table).\\
	
	Both dKL and MCKL values grow with increasing absolute value of $\tau$ as we would expect from the true KL distance. As before, the rank correlation between dKL and MCKL values is equal to 1, the ratio is nearly constant and the dKL is computed 40 times faster than the MCKL.
	
	In the third plausibility check we consider five-dimensional Gumbel vines (i.e.\ every pair-copula is a bivariate Gumbel copula having upper tail dependence) with the same structure matrix $D_d$ and Kendall's $\tau$ matrix $K_5(0.5)$. In \autoref{tab:dKL_Plausi1} (right table) we compare this vine to its nearest Gaussian vine and other vines constructed similarly using one copula family only but retaining the same dependence in terms of the Kendall's $\tau$ matrix. As other copula families we choose the Clayton copula ($\Cc$) exhibiting lower tail dependence, the survival Clayton copula ($s\Cc$) with upper tail dependence and the Joe copula ($\Jc$) having lower tail dependence.    As the difference between upper and lower tail dependent pair-copulas is large, we expect the highest distance value for the Clayton vine. Conversely, the distance to the survival Clayton vine should be the lowest.
	The diagonal KL distance also passes this plausibility check assigning the largest distance to the Clayton vine, a small distance to the survival Clayton vine and medium distances to the Joe and Gaussian vines. Again, the ratio between dKL and MCKL values varies only little and the dKL is still roughly 40 times faster than the MCKL regarding computational time.
\end{exa}
The previous plausibility checks have shown that the diagonal Kullback-Leibler distance is a reasonable and fast distance measure for five-dimensional vines. Since the main motivation of the reduction to diagonals was reduced computational complexity, we now turn to higher dimensional examples.

\begin{exa}[Performance in different dimensions]\label{exa:dim_ana1}
	In order to assess the performance of the diagonal KL distance regarding computational time in different dimensions, we again make use of the Euro Stoxx 50 data. We take the 12 German stocks (with ticker symbols ALV, BAS, BAYN, DAI, DB1, DBK, DTE, EOAN, MUV2, RWE, SIE, SAP, corresponding to ${U_1,\ldots,U_{12}}$) 
	and fit vines to the first $d$ variables ($d=3,\ldots,12$). We display the dKL distance ($n=10$) between these vines and their nearest Gaussian vines with corresponding computational times (in seconds) in \autoref{tab:ESdKL_dim}. Again, we also present the approximated KL values using Monte Carlo integration ($N_{\text{MC}}=10^6$) and the ratio between dKL and MCKL values.
	
	While we observe that the dKL is exceptionally fast in low dimensions we note that computational times more than double when moving up one dimension. This is reasonable since the total number of diagonals in all evaluation grids is equal to
	\begin{equation}
	\sum_{j=1}^{d-1}|\Dc_j^u|=\sum_{j=1}^{d-1}2^{d-j-1}=2^{d-1}-1
	\end{equation}
	and thus grows exponentially in $d$. Further, the evaluations of the conditional copula densities become more costly in higher dimensions, which can also be seen by the fact that the computational times for the MCKL increase even though the number of evaluations $N_{\text{MC}}$ stays constant. Comparing the computational times of the dKL and MCKL one notices that the dKL is considerably faster than the MCKL in lower dimensions. Only in dimensions 10 and higher it looses this competitive advantage. Considering the ratio between dKL and MCKL values it seems that the dKL values increase slightly faster with the dimension $d$ than the MCKL values.
		\begin{table*}[t]
			\centering
			\begin{tabular}{|l|cccccccccc|}
				\hline
				$d$ & 3 & 4 & 5 & 6 & 7 & 8 & 9 & 10 & 11 & 12 \\
				\hline
				dKL & 0.076 & 0.109 & 0.178 & 0.249 & 0.297 & 0.459 & 0.529 & 0.657 & 0.670 & 0.839 \\
				time [s] & 0.4 & 1 & 3 & 8 & 22 & 61 & 145 & 342 & 788 & 1846 \\
				\hline
				MCKL & 0.048 & 0.075 & 0.098 & 0.140 & 0.172 & 0.211 & 0.240 & 0.287 & 0.322 & 0.354 \\
				time [s] & 25 & 43 & 87 & 129 & 158 & 174 & 235 & 279 & 408 & 448 \\
				\hline
				ratio & 1.57 & 1.45 & 1.81 & 1.78 & 1.73 & 2.17 & 2.2 & 2.29 & 2.08 & 2.37 \\
				\hline
			\end{tabular}
			\caption{dKL values ($n=10$) and MCKL values ($N_{\text{MC}}=10^6$) with corresponding computational times (in seconds) for vines with different dimensions based on the stock exchange data. The fifth row contains the ratio between dKL and MCKL.}
			\label{tab:ESdKL_dim}
		\end{table*}
		
\end{exa}
The preceding examples suggest that with the dKL distance we have found a valid distance measure between vines with reasonable computational times for up to ten dimensions. However, we are still interested in finding a distance measure computable in dimensions of order 30 to 50. To achieve this, the number of grid points should not depend on the dimension of the evaluation grid, implying a constant number of grid points. Hence, we choose only one of the $2^{d-j-1}$ warped discrete diagonals in $\Dc_j^u$ to be the evaluation grid. While this may seem like a very severe restriction (with  the curse of dimensionality in mind), two heuristic observations justify this approach. On the one hand we observe that most of the $2^{d-j-1}$ diagonals contain many grid points with density values close to zero while there is always one diagonal whose points have very large density values. On the other hand we will see that the properties of the distance measure using only this single diagonal for the evaluation grid still pass the plausibility checks with values behaving closely to those of the dKL and the MCKL.

\subsection{Single diagonal Kullback-Leibler distance}\label{subsec:sdKL}

In order to find the one diagonal whose grid points have the highest density values we introduce a \emph{weighting measure} that assigns a positive real number to a diagonal depending on how the density behaves on it. The higher the density values are the more weight the corresponding diagonal obtains.
\begin{defi}[Diagonal weighting measure]\label{defi:dwm}
	Assume we can parametrize a diagonal $D\subseteq [0,1]^{d}$ (on the u-scale) by the mapping $\gammab\colon [0,1]\to[0,1]^d$. Let $c\colon[0,1]^{d}\rightarrow\left[0,\infty\right)$ be a copula density. Then, we define
	\begin{equation}\label{eq:dwm}
	\lambda_c(D):=\int_{\xib\in D}c(\xib)\diff\xib=\int_{t\in[0,1]}c(\gammab(t))\left\|\dot{\gammab}(t)\right\|\diff t
	\end{equation}
	to be the \emph{weight of $D$ under $c$}, where $\dot{\gammab}$ is the vector of componentwise derivatives of $\gammab$.
\end{defi}

We now define the \emph{single diagonal Kullback-Leibler distance}, which is a version of the diagonal Kullback-Leibler distance that only evaluates the diagonal with the highest weight.
\begin{defi}[Single diagonal Kullback-Leibler distance]\label{defi:sdKL}
	Let $\Rc^f$, $\Rc^g$ be as before and $\Dc_{j,k^*_j}^u$ with $|\Dc_{j,k^*_j}^u|=n$ and $k^*_j:=\argmax_k{\lambda_{c^f_{(j+1):d}}}\big(D^u_{j,k}\big)$ be a discretization of the corresponding diagonal with the highest weight according to $\lambda_{c^f_{(j+1):d}}$,
	${j=1,\ldots,d-1}$. Then, the \emph{single diagonal Kullback-Leibler distance} ($\sdKL$) between $\Rc^{f}$ and $\Rc^{g}$ is defined as
	\begin{equation*}
	\sdKL\big(\Rc^{f},\Rc^{g}\big):=\sum_{j=1}^{d-1} \frac{1}{\big|\Dc_{j,k^*_j}^u\big|} \sum_{\ub_{(j+1):d}\in \Dc_{j,k^*_j}^u} \KL\left(c^{f}_{j|(j+1):d}\left(\,\cdot\,|\ub_{(j+1):d}\right),c^{g}_{j|(j+1):d}\left(\,\cdot\,|\ub_{(j+1):d}\right)\right).
	\end{equation*}
\end{defi}

\begin{rem}
	From \autoref{rem:dKL} we know that the single diagonal Kullback-Leibler distance approximates a scaled line integral over weighted univariate KL distances between conditional densities along the diagonal with the highest weight.
	
	To find this diagonal we actually would have to calculate the integral of $c^f_{(j+1):d}$ over each of the $2^{d-j-1}$ diagonals. In practice, this may be infeasible for high dimensions. Therefore, we propose a more sophisticated method to find a candidate for the diagonal with the highest weight. Similar to the hill-climbing algorithm used to find optimal graph structures in Bayesian networks \citep[see][]{tsamardinos2006max}, we choose a starting value in form of a certain diagonal implied by the vine's unconditional dependencies and look in the ``neighborhood'' of this diagonal for another diagonal with higher weight. This procedure is repeated until a (local) maximum is found. The two procedures of finding a suitable starting diagonal and locally searching for better candidates are described in \autoref{App:weight}.
\end{rem}

In the following we continue the plausibility checks from \autoref{exa:dKL_Plausi} to demonstrate that the restriction to a single diagonal is still enough for the resulting distance measure to generate reasonable results.

\begin{exa}[Plausibility checks (\autoref{exa:dKL_Plausi} continued)]\label{exa:sdKL_Plausi}
	
	\autoref{tab:sdKL_Plausi1} repeats the three plausibility checks of \autoref{exa:dKL_Plausi}. The resulting sdKL values obviously also pass these tests resembling the behavior of the MCKL values quite closely with relatively steady sdKL/MCKL-ratios. Evaluating at only one diagonal in each grid reduces computational times even more such that the sdKL is roughly 180 times faster than the MCKL.
	\begin{table*}[!h]
		\centering
		\begin{tabular}{|r|r|r|r|}
			\hline
			$\nu$ & sdKL & MCKL & ratio \\
			\hline
			3 & 0.754 & 0.374 & 2.02 \\
			5 & 0.330 & 0.162 & 2.04 \\
			7 & 0.184 & 0.091 & 2.03 \\
			10 & 0.097 & 0.047 & 2.06 \\
			15 & 0.046 & 0.023 & 2.04 \\
			20 & 0.026 & 0.013 & 1.98 \\
			25 & 0.017 & 0.008 & 2.04 \\
			30 & 0.012 & 0.006 & 2.03 \\
			\hline
		\end{tabular}
		\,
		\begin{tabular}{|r|r|r|r|}
			\hline
			$\tau$ & sdKL & MCKL & ratio \\
			\hline
			$-0.7$ & 7.534 & 3.226 & 2.34 \\
			$-0.5$ & 3.100 & 1.431 & 2.17 \\
			$-0.3$ & 0.773 & 0.473 & 1.63 \\
			$-0.1$ & 0.053 & 0.050 & 1.05 \\
			0.1 & 0.157 & 0.048 & 3.30 \\
			0.3 & 1.322 & 0.423 & 3.12 \\
			0.5 & 3.226 & 1.262 & 2.56 \\
			0.7 & 6.193 & 2.982 & 2.08 \\
			\hline
		\end{tabular}
		\,
		\begin{tabular}{|r|r|r|r|}
			\hline
			family & sdKL & MCKL & ratio \\
			\hline
			$\Nc$ & 0.394 & 0.205 & 1.92 \\
			$\Cc$ & 3.557 & 1.780 & 2.00 \\
			$s\Cc$ & 0.421 & 0.158 & 2.66 \\
			$\Jc$ & 0.576 & 0.249 & 2.32 \\
			\hline
		\end{tabular}
		\caption{Left table: sdKL ($n=10$) and MCKL ($N_{\text{MC}}=10^6$) values between five-dimensional t copulas with $P^{(2)}=V_5(\nu)$ and their nearest Gaussian vines. Middle table: sdKL and MCKL values between five-dimensional t copulas with $K=K_5(\tau)$ and their nearest Gaussian vines. Right table: sdKL and MCKL values between a five-dimensional Gumbel vine and vines constructed using one copula family only. The third column contains the ratio between sdKL and MCKL.}
		\label{tab:sdKL_Plausi1}
	\end{table*}
	
\end{exa}
These five-dimensional plausibility checks empirically show that the reduction from all to one diagonal still yields viable results for our modified version of the KL distance. As a final application we want to compare all distance measures introduced in this paper.

\subsection{Comparison of all introduced distance measures}
\begin{table*}[t]
	\centering
	\begin{tabular}{|l|rrrrrrrr|}
		\hline
		$d$ & 3 & 4 & 5 & 7 & 10 & 15 & 20 & 30 \\
		\hline
		aKL & 98.4 & 98.5 & 98.4 & -- & -- & -- & -- & -- \\
		dKL & 96.7 & 97.4 & 98.7 & 98.7 & 97.2 & -- & -- & -- \\
		sdKL & 93.3 & 90.2 & 91.5 & 89.7 & 82.9 & 84.8 & 84.5 & 80.4 \\
		MCKL & 99.8 & 99.5 & 99.5 & 99.7 & 98.5 & 97.2 & 92.3 & 91.7 \\
		\hline
	\end{tabular}
	\caption{Rank correlations (in percent) between the true Kullback-Leibler distance and the results of aKL, dKL, sdKL and MCKL, respectively, for dimensions $d=3,4,5,7,10,15,20,30$.}
	\label{tab:Gaussvines}
\end{table*}
\begin{table*}[t]
	\centering
	\begin{tabular}{|l|rrrrrrrr|}
		\hline
		$d$  & 3 & 4 & 5 & 7 & 10 & 15 & 20 & 30 \\
		\hline
		aKL & 3.46 & 117.67 & 4357.91 & -- & -- & -- & -- & -- \\
		dKL & 0.23 & 0.82 & 2.49 & 19.63 & 338.32 & -- & -- & -- \\
		sdKL & 0.18 & 0.38 & 0.69 & 1.80 & 4.86 & 16.04 & 36.91 & 114.12 \\
		MCKL & 7.35 & 14.50 & 24.12 & 46.05 & 97.54 & 239.17 & 473.41 & 961.12 \\
		\hline
	\end{tabular}
	\caption{Average computational times (in seconds) of aKL, dKL, sdKL and MCKL for dimensions $d=3,4,5,7,10,15,20,30$.}
	\label{tab:Gaussvines_times}
\end{table*}

In the following example, we will investigate the behavior of the KL, aKL, dKL, sdKL and MCKL in dimensions $d=3,4,5,7,10,15,20,30$. We make use of the fact that the Kullback-Leibler distance between Gaussian copulas can be expressed analytically (\cite{hershey2007approximating}). For two Gaussian copulas $c^f$ and $c^g$ with correlation matrices $\Sigma^f$ and $\Sigma^g$, respectively, one has
\[
\KL\big(c^f,c^g\big)=\frac{1}{2}\left\lbrace \ln\left(\frac{\det\left(\Sigma^g\right)}{\det\left(\Sigma^f\right)}\right)+\tr\left( (\Sigma^g)^{-1}\Sigma^f\right)-d\right\rbrace ,
\]
where $\det (\,\cdot\,)$ denotes the determinant and $\tr(\,\cdot\,)$ the trace of a matrix. For each dimension $d$ we use a reference Gaussian vine $\Rc^0$ (which is also a Gaussian copula) with the (D-vine) structure matrix $M=D_d$ and Kendall's $\tau$ matrix $K=K_d(0.5)$ (cf.\ \autoref{eq:D_d} and \autoref{eq:K_d}, respectively). \\
We generate another $m=50$ Gaussian vines $\Rc^r$, $r=1,\ldots,m$, with the same structure matrix $M=D_d$ and a parameter matrix $P^{(1)}$, where the $d(d-1)/2$ partial correlations are simulated such that the corresponding correlation matrix is uniform over the space of valid correlation matrices. For this purpose, we follow \cite{joe2006generating}: For $i=2,\ldots,d$ and $j=1,\ldots,i-1$ we draw $p^{(1)}_{i,j}$ from a $\text{Beta}(i/2,i/2)$ distribution and transform it linearly to $[-1,1]$.\\
We compare $\Rc^0$ to each $\Rc^r$ using the model distances KL, MCKL, aKL, dKL and sdKL. Since the Kullback-Leibler distance is exact in these cases, we can assess the performance of the remaining distance measures by comparing their $m=50$ distance values to the ones of the true KL. As the scale of the KL and related distance measures cannot be interpreted in a sensible way, we are only interested in how well the ordering suggested by the KL is reproduced by aKL ($n=20$), dKL ($n=10$), sdKL ($n=10$) and MCKL ($N_{\text{MC}}=10^6$), respectively. Hence, we consider the respective rank correlations to the KL values in order to assess their performances. The results and average computation times are displayed in \autoref{tab:Gaussvines} and \autoref{tab:Gaussvines_times}, respectively. \\
With a rank correlation of more than $98\%$ the approximate KL performs extremely well for $d=3,4,5$. However, computational times increase drastically with the dimension such that it cannot be computed in higher dimensions in a reasonable amount of time. As the plausibility checks from the previous sections suggested the diagonal KL also produces very good results. In lower dimensions the dKL is competitive regarding computational times. Only for dimensions 10 and higher it becomes slower due to the exponentially increasing number of diagonals. Therefore, calculations have not been performed for $d=15,20,30$. As expected, restricting to only one diagonal reduces computational times considerably such that even in very high dimensions they are kept to a minimum. Of course, this restriction comes along with slight loss of performance, still achieving a rank correlation of over $80\%$ in 30 dimensions. Being a consistent estimator of the KL distance (for $N_{\text{MC}}\to \infty$), the Monte-Carlo KL has the best performance of the considered model distances. However, the performance decreases for high dimensions due to the curse of dimensionality ($N_{\text{MC}}=10^6$ for all $d$). Further, the price of the slightly better performance (compared to sdKL) is a considerably higher computational time, e.g.\ in 10 and 30 dimensions the sdKL is roughly 20 and 9 times faster than the MCKL, respectively.

Altogether we can say that in order to have good performance and low computational times one should use the dKL in lower dimensions and then switch to the sdKL in higher dimensions in order to obtain a usable proxy for the KL distance at (relatively) low computational costs.

\section{Conclusion}\label{sec:conclusion}
In this paper we have developed new methods for measuring model distances between vine copulas. Since vines are frequently used for high dimensional dependence modeling, the focus was to propose concepts that can in particular be applied to higher dimensional models. With the approximate Kullback-Leibler distance we introduced a measure which converges to the original Kullback-Leibler distance and therefore produces good approximations. Although being considerably faster than the calculation of the KL by numerical integration, the aKL suffers from the curse of dimensionality and therefore is not computationally tractable in dimensions $d\geq6$. Being a more crude approximation the diagonal Kullback-Leibler distance, which highlights the difference between vines conditioned on points on the diagonals, has proven itself to be a reliable and computationally parsimonious model distance measure for comparing vines up to 10 dimensions. In higher dimensions the number of diagonals becomes intractable, which is why we suggested to reduce calculations to only one diagonal with large density values, introducing the single diagonal Kullback-Leibler distance. With the sdKL, we have found a possibility to overcome the shortfalls of alternative methods like Monte Carlo (low speed and randomness) and at the same time maintain the desired properties of the Kullback-Leibler distance relatively well.\\
In ongoing research we address ourselves to applying our distance measures in many possible fields: applications to real data sets, comparing C- and D-vines, determining appropriate truncation levels for vines and comparing non-simplified vines with simplified ones in order to develop a test to decide whether the simplifying assumption is satisfied for given data.

\section*{Appendix}
\appendix

\section{Proof of \autoref{prop:column}}\label{App:column}
From \autoref{eq:vinepdfdecomp} we know that the vine copula density can be written as a product over the pair-copula expressions corresponding to the matrix entries. In Property 2.8 (ii), \cite{dissmann2013selecting} state that deleting the first row and column from a $d$-dimensional structure matrix yields
a $(d-1)$-dimensional trimmed structure matrix. Due to Property 2 from \autoref{defi:RVM} the entry $m_{1,1}=1$ does not appear in the remaining matrix. Hence, we obtain the density $c_{2:d}$ by taking the product over all pair-copula expressions corresponding to the entries in the trimmed matrix. Iterating this argument yields that the entries of matrix $M_k:=(m_{i,j})_{i,j=k+1}^d$ resulting from cutting the first $k$ rows and columns from $M$ represent the density $c_{(k+1):d}$. In general, we have
\[
c_{j|(j+1):d}\left(u_j|u_{j+1},\ldots,u_d\right)=\frac{c_{j:d}\left(u_j,\ldots,u_d\right)}{c_{(j+1):d}\left(u_{j+1},\ldots,u_d\right)}.
\]
The numerator and denominator can be obtained as the product over all pair-copula expressions corresponding to the entries of $M_{j-1}$ and $M_j$. Thus, $c_{j|(j+1):d}$ is simply the product over the expressions corresponding to the entries from the first column of $M_{j-1}$. This proves \autoref{eq:conddens}.\hfill $\square$
\section{Proof of \autoref{prop:KLdecomp}} \label{App:KLdecomp}

Recall that using recursive conditioning we can obtain
\[
c^{i}\left(u_1,\ldots,u_d\right) =\prod_{j=1}^d{c^{i}_{j|(j+1):d}}\left(u_j|\ub_{(j+1):d}\right),\quad i\in\left\{f,g\right\}.
\]
Thus, the Kullback-Leibler distance between $c^{f}$ and $c^{g}$ can be written in the following way:
\begin{equation*}
\begin{split}
\KL\big(c^{f},c^{g}\big)&=\int\limits_{\ub\in\left[0,1\right]^d}{\ln\left(\frac{c^{f}(\ub)}{c^{g}(\ub)}\right)c^{f}(\ub)\diff\ub}\\
&=\int\limits_{\ub\in\left[0,1\right]^d}\sum_{j=1}^{d}{\ln\left(\frac{c^{f}_{j|(j+1):d}\left(u_j|\ub_{(j+1):d}\right)}{c^{g}_{j|(j+1):d}\left(u_j|\ub_{(j+1):d}\right)}\right)c^{f}(\ub)\diff\ub}\\
&=\sum_{j=1}^{d}\int\limits_{u_d\in\left[0,1\right]}\!\!\cdots\!\!\int\limits_{u_1\in\left[0,1\right]}{\ln\left(\frac{c^{f}_{j|(j+1):d}\left(u_j|\ub_{(j+1):d}\right)}{c^{g}_{j|(j+1):d}\left(u_j|\ub_{(j+1):d}\right)}\right)c^{f}(u_1,\ldots,u_d)\diff u_1\cdots\ \diff u_d}\\
&=\sum_{j=1}^{d}\int\limits_{u_d\in\left[0,1\right]}\!\!\cdots\!\!\int\limits_{u_j\in\left[0,1\right]}\ln\left(\frac{c^{f}_{j|(j+1):d}\left(u_j|\ub_{(j+1):d}\right)}{c^{g}_{j|(j+1):d}\left(u_j|\ub_{(j+1):d}\right)}\right)\\
&\quad \times\left\{\int\limits_{u_{j-1}\in\left[0,1\right]}\!\!\cdots\!\!\int\limits_{u_1\in\left[0,1\right]}c^{f}(u_1,\ldots,u_d)\diff u_1\cdots \diff u_{j-1}\right\}\diff u_j\cdots\ \diff u_d\\
&=\sum_{j=1}^{d}\int\limits_{u_d\in\left[0,1\right]}\!\!\cdots\!\!\int\limits_{u_j\in\left[0,1\right]}\ln\left(\frac{c^{f}_{j|(j+1):d}\left(u_j|\ub_{(j+1):d}\right)}{c^{g}_{j|(j+1):d}\left(u_j|\ub_{(j+1):d}\right)}\right)c^{f}_{j,\ldots,d}(u_j,\ldots,u_d)\diff u_j\cdots\ \diff u_d\\
&=\sum_{j=1}^{d}\int\limits_{u_d\in\left[0,1\right]}\!\!\cdots\!\!\int\limits_{u_{j+1}\in\left[0,1\right]}\left\{\int\limits_{u_{j}\in\left[0,1\right]}\ln\left(\frac{c^{f}_{j|(j+1):d}\left(u_j|\ub_{(j+1):d}\right)}{c^{g}_{j|(j+1):d}\left(u_j|\ub_{(j+1):d}\right)}\right)c^{f}_{j|(j+1):d}(u_j|\ub_{(j+1):d})\diff u_j\right\}\\
&\quad \times c^{f}_{(j+1):d}(\ub_{(j+1):d})\diff u_{j+1}\cdots \diff u_d\\
&=\sum_{j=1}^d\Ebb_{c^{f}_{(j+1):d}}\left[\KL\left(c^{f}_{j|(j+1):d}(\,\cdot\,|\Ub_{(j+1):d}),c^{g}_{j|(j+1):d}(\,\cdot\,|\Ub_{(j+1):d})\right)\right]. 
\end{split} 
\end{equation*}
\hfill $\square$

\section{Number of vines with the same diagonal} \label{App:same_diag}
\begin{prop}\label{prop:numbervines}
	Let $\sigmab=\left(\sigma_1,\ldots,\sigma_d\right)'$ be a permutation of ${1\!:\!d}$. Then, there exist $2^{\binom{d-2}{2}+d-2}$ different vine decompositions whose structure matrix has the diagonal $\sigmab$.
\end{prop}
\begin{proof}
	The number of vine decompositions whose structure matrix has the same diagonal $\sigma$ can be calculated as the the quotient of the number of valid structure matrices and the number of possible diagonals. \cite{morales2011count} show that there are $\frac{d!}{2}\cdot 2^{\binom{d-2}{d}}$ different vine decompositions. In each of the $d-1$ steps of the algorithm for encoding a vine decomposition in a structure matrix \citep[see][]{stoeber2012} we have two possible choices such that there are $2^{d-1}$ structure matrices representing the same vine decomposition. Hence, there are in total $\frac{d!}{2}\cdot 2^{\binom{d-2}{d}}\cdot 2^{d-1}$ valid structure matrices. Further, there are $d!$ different diagonals. Thus, for a fixed diagonal $\sigmab$ there exist 
	\[
	\frac{\frac{d!}{2}\cdot 2^{\binom{d-2}{d}}\cdot 2^{d-1}}{d!}=2^{\binom{d-2}{2}+d-2}
	\]
	different vine decompositions.
\end{proof}

\section{Proof of \autoref{prop:aKL_limit}} \label{App:aKL_limit}
Let $\eps>0$ and $n\in\Nbb$. To simplify notation, for $j=1,\ldots,d-1$ we define
\[
\kappa_j\left(\ub_{(j+1):d}\right):=\KL\left(c^{f}_{j|(j+1):d}\left(\,\cdot\,|\ub_{(j+1):d}\right),c^{g}_{j|(j+1):d}\left(\,\cdot\,|\ub_{(j+1):d}\right)\right).
\]
Then, by definition
\[
\begin{split}
\aKL\big(\Rc^f,\Rc^g\big)=\sum_{j=1}^{d-1}\frac{1}{n^{d-j}}\sum_{\ub_{(j+1):d}\in\Gc_j}\kappa_j\left(\ub_{(j+1):d}\right)=\sum_{j=1}^{d-1}\frac{1}{n^{d-j}}\sum_{\wb_{(j+1):d}\in\Wc_j}\kappa_j\big(T_{c_{(j+1):d}^f}(\wb_{(j+1):d})\big).
\end{split}
\]
Since $W_j$ is a discretization of $[\eps,1-\eps]^{d-j}$ with mesh going to zero for $n\to \infty$, we have
\begin{equation*}
\frac{1}{n^{d-j}}\sum_{\wb_{(j+1):d}\in\Wc_j}\kappa_j\big(T_{c_{(j+1):d}^f}(\wb_{(j+1):d})\big)\stackrel{n\to\infty}{\longrightarrow}\int_{[\eps,1-\eps]^{d-j}}\kappa_j\big(T_{c_{(j+1):d}^f}(\wb_{(j+1):d})\big)\diff \wb_{(j+1):d}.
\end{equation*}
Substituting $\wb_{(j+1):d}=T^{-1}_{c_{(j+1):d}^f}\left(\ub_{(j+1):d}\right)$ yields
\begin{equation*}
\int\limits_{[\eps,1-\eps]^{d-j}}\kappa_j\big(T_{c_{(j+1):d}^f}(\wb_{(j+1):d})\big)\diff \wb_{(j+1):d}=\!\!\int\limits_{T_{c_{(j+1):d}}^f\left([\eps,1-\eps]^{d-j}\right)}\!\!\!\!\kappa_j\big(\ub_{(j+1):d}\big) c_{(j+1):d}^f\left(\ub_{(j+1):d}\right)\diff \ub_{(j+1):d}
\end{equation*}
since 
\[
T^{-1}_{c_{(j+1):d}^f}\left(\ub_{(j+1):d}\right)=\big(C^f_{j+1|(j+2):d}(u_{j+1}|\ub_{(j+2):d}), \ldots, C^f_{d-1|d}(u_{d-1}|\ub_{d}), u_d\big)'
\] 
with (upper triangular) Jacobian matrix
\[
J=J_{T^{-1}_{c_{(j+1):d}^f}}\mkern-18mu\left(\ub_{(j+1):d}\right)=\begin{pmatrix}
c^f_{j+1|(j+2):d}(u_{j+1}|\ub_{(j+2):d}) & &  & \\
& \ddots & {\LARGE\text{*}} & \\
{\Large\text{0}} & & c_{d-1|d}^f(u_{d-1}|\ub_{d}) & \\
& & & 1
\end{pmatrix}
\]
such that $\diff \wb_{(j+1):d}= \det(J) \diff \ub_{(j+1):d}=c_{(j+1):d}^f\left(\ub_{(j+1):d}\right)\diff \ub_{(j+1):d}$. Since we are only interested in the determinant of $J$, whose lower triangular matrix contains only zeros, the values in the upper triangular matrix (denoted by $\Ast$) are irrelevant here. Finally, using the fact that
\[
\lim_{\eps\to 0} T_{c_{(j+1):d}^f}\left([\eps,1-\eps]^{d-j}\right)=T_{c_{(j+1):d}^f}\left([0,1]^{d-j}\right)=[0,1]^{d-j},
\] we obtain
\[
\begin{split}
\lim_{\eps\to 0}\lim_{n\to\infty}\aKL\big(\Rc^f,\Rc^g\big) = \sum_{j=1}^{d-1} \int\limits_{[0,1]^{d-j}}\!\!\!\!\kappa_j\big(\ub_{(j+1):d}\big) c_{(j+1):d}^f\left(\ub_{(j+1):d}\right)\diff \ub_{(j+1):d}\stackrel{\text{\textcolor{TUMblue}{Prop. }}\ref{prop:KLdecomp}}{=}\KL\big(c^f,c^g\big).
\end{split}
\]
\hfill $\square$

\section{Regarding \autoref{rem:dKL}}\label{App:remark}
\subsection{Limit of the dKL} \label{App:dKL_approx}
Let $\eps>0$ and $n\in\Nbb$. Again, for $j=1,\ldots,d-1$ we define
\[
\kappa_j\left(\ub_{(j+1):d}\right):=\KL\left(c^{f}_{j|(j+1):d}\left(\,\cdot\,|\ub_{(j+1):d}\right),c^{g}_{j|(j+1):d}\left(\,\cdot\,|\ub_{(j+1):d}\right)\right).
\]
The contribution of $\Dc^u_{j,k}$, $j=1,\ldots,d-1$, $k=1,\ldots,2^{d-j-1}$, to the dKL is given by
\[
\begin{split}
\frac{1}{n}\sum_{\ub_{(j+1):d}\in \Dc^u_{j,k}}\kappa_j\left(\ub_{(j+1):d}\right)=\frac{1}{n}\sum_{\wb_{(j+1):d}\in \Dc^w_{j,k}}\kappa_j\big(T_{c_{(j+1):d}^f}(\wb_{(j+1):d})\big)=\frac{1}{n}\sum_{i=1}^n\kappa_j\big(T_{c_{(j+1):d}^f}(\omegab(t_i))\big), 
\end{split}
\]
where $\omegab(t)=\rb+t \vb(\rb)$ with $\rb\in \left\lbrace 0,1\right\rbrace^{d-j}$ being a corner point of $D^w_{j,k}$ and $t_i=\eps+(i-1)\frac{1-2\eps}{n-1}$ for $i=1,\ldots,n$. Letting $n\to\infty$ yields
\begin{equation}\label{eq:lim_int}
\frac{1}{n}\sum_{i=1}^n\kappa_j\big(T_{c_{(j+1):d}^f}(\omegab(t_i))\big)\stackrel{n\to\infty}{\longrightarrow}\int_{t\in[\eps,1-\eps]}\kappa_j\big(T_{c_{(j+1):d}^f}(\omegab(t))\big)\diff t.
\end{equation}
Now, we further let $\eps\to0$ and use the fact that $\left\| \dot{\omegab}(t)\right\|=\sqrt{d-j}$ to obtain
\[
\begin{split}
\int\limits_{t\in[0,1]}\kappa_j\big(T_{c_{(j+1):d}^f}(\omegab(t))\big)\diff t&=\frac{1}{\sqrt{d-j}}\int\limits_{t\in[0,1]}\kappa_j\big(T_{c_{(j+1):d}^f}(\omegab(t))\big)\left\| \dot{\omegab}(t)\right\|\diff t\\
&=\frac{1}{\sqrt{d-j}}\int\limits_{\wb_{(j+1):d}\in D^w_{j,k}}\kappa_j\big(T_{c_{(j+1):d}^f}(\wb_{(j+1):d})\big)\diff\wb_{(j+1):d}\\
&=\frac{1}{\sqrt{d-j}}\int\limits_{\ub_{(j+1):d}\in D^u_{j,k}}\kappa_j\big(\ub_{(j+1):d}\big)c_{(j+1):d}\big(\ub_{(j+1):d}\big)\diff\ub_{(j+1):d},
\end{split}
\]
where we used the substitution $\ub_{(j+1):d}:=T^{-1}_{c_{(j+1):d}^f}(\wb_{(j+1):d})$, $\diff\wb_{(j+1):d}=c_{(j+1):d}^f(\ub_{(j+1):d})\diff \ub_{(j+1):d}$ (cf.\ \autoref{App:aKL_limit}) in the last line. \hfill $\square$

\subsection{Tail transformation} \label{App:tail_trafo}
In our empirical applications of the dKL, we have noticed that different vines tend to differ most in the tails of the distribution. Therefore, we increase the concentration of evaluation points in the tails of the diagonal by transforming the points $t_i$, $i=1,\ldots,n$, via a suited function $\Psi$. Hence, by substituting $t=\Psi(s)$ in \autoref{eq:lim_int} we obtain
\[
\int_{s\in\Psi^{-1}([\eps,1-\eps])}\kappa_j\Big(T_{c_{(j+1):d}^f}\big(\eta\big(\Psi(s)\big)\big)\Big)\Psi'(s)\diff s.
\]
We use its discrete pendant
\[
\frac1n\sum_{i=1}^n\kappa_j\Big(T_{c_{(j+1):d}^f}\big(\eta\big(\Psi(s_i)\big)\big)\Big)\Psi'(s_i),
\]
where $s_i=\Psi^{-1}(\eps)+(i-1)\frac{\Psi^{-1}(1-\eps)-\Psi^{-1}(\eps)}{n-1}$ for $i=1,\ldots,n$. Regarding the choice of $\Psi$, all results in this paper are obtained using
\[
\Psi_a\colon [0,1]\to [0,1],\quad \Psi_a(t):=\frac{\Phi(2a(t-0.5))-\Phi(-a)}{2\Phi(a)-1}
\]
with shape parameter $a>0$, where $\Phi$ is the standard normal distribution function.
\autoref{fig:prob_trafo} shows the graph of $\Psi_a$ for different values of $a$. We see that larger values of $a$ imply more points being transformed into the tails. Having tested different values for $a$, we found that $a=4$ yields the best overall results. Therefore, we consistently use $a=4$. 
\begin{figure}[!htbp]
	\centering
	\includegraphics[trim=0.3cm 0.6cm 0.3cm 2cm,clip,width=0.5\textwidth]{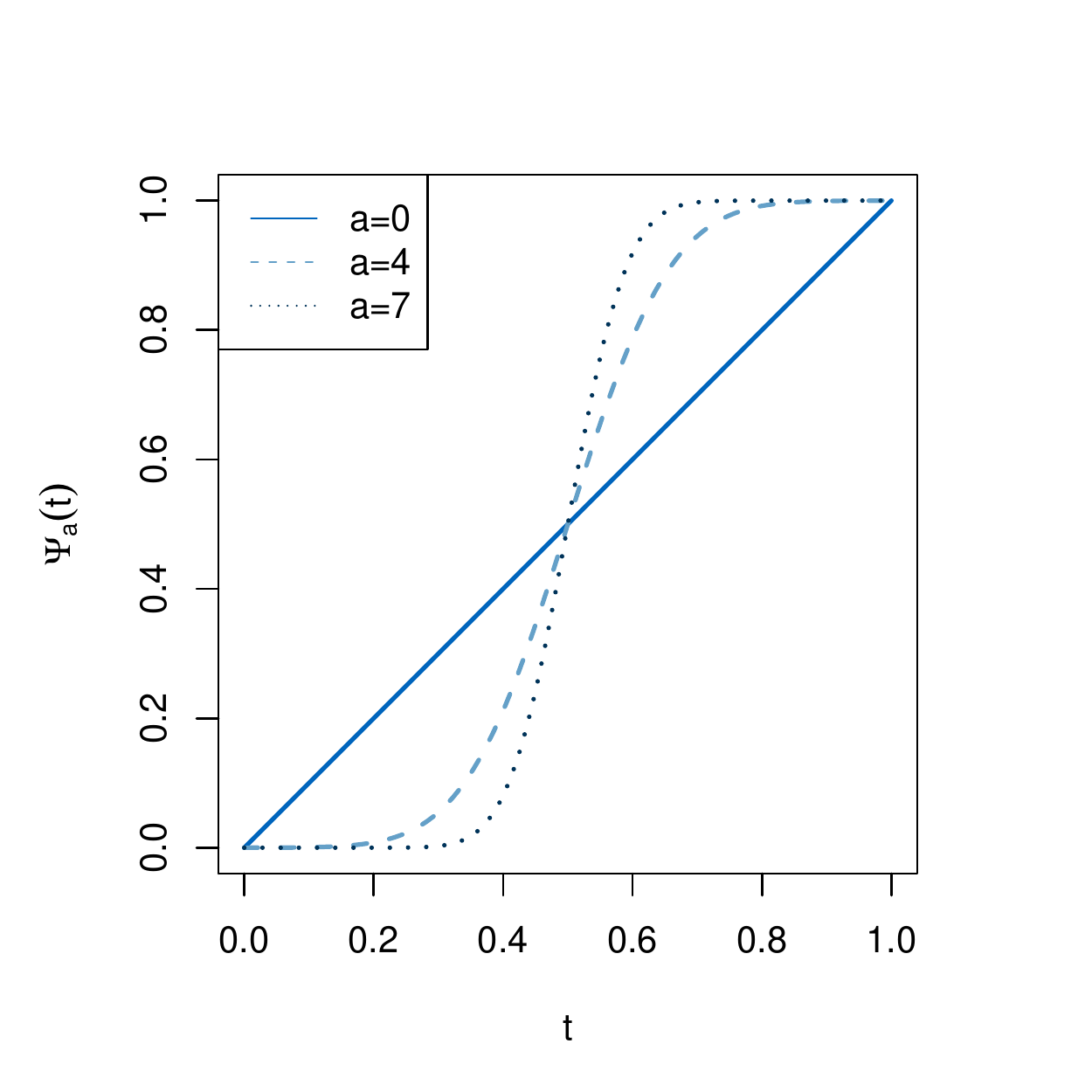}
	\caption{Plot of $\Psi_a$ for $a=0,4,7$.}
	\label{fig:prob_trafo}
\end{figure}

\section{Finding the diagonal with the highest weight} \label{App:weight}
\subsection{Procedure 1: Finding a starting value}\label{App:start}

The idea behind the following heuristic is that a diagonal has a higher weight if its points have high probability implied by the copula density. Hence, the diagonal should reflect the dependence structure of the variables. The unconditional dependence in a vine captures most of the total dependence and is easy to interpret. For example, if $U_i$ and $U_j$ are positively dependent (i.e.\ $\tau_{i,j}>0$) and $U_j$ and $U_k$ are negatively dependent (i.e.\ $\tau_{j,k}<0$), then it seems plausible that $U_i$ and $U_k$ are negatively dependent. This concept can be extended to arbitrary dimensions.
\begin{enumerate}
	\item Take each variable to be a node in an empty graph.
	\item Consider the last row of the structure matrix, encoding the unconditional pair-copulas. Connect two nodes by an edge if the dependence of the corresponding variables is described by one of those copulas.
	\item Assign a ``+" to node 1.
	\item As long as not all nodes have been assigned a sign, repeat:
	\begin{enumerate}
		\item For each node that has been assigned a sign in the previous step, consider its neighborhood.
		\item If the root node has a ``+", then assign to the neighbor node the sign of the Kendall's $\tau$ of the pair-copula connecting the root and neighbor node, else the opposite sign.		 
	\end{enumerate}
	\item The resulting direction vector $\vb\in \left\lbrace -1,1\right\rbrace ^d$ has entries $v_i$ which are $1$ or $-1$ if node $i$ is has been assigned a ``$+$" or a ``$-$", respectively.
\end{enumerate}
Note that if we had assigned a ``$-$'' to node 1 in step 3, we would have ended up with $-v$ instead of $v$, implying the same diagonal.

To illustrate the procedure from above we consider a nine-dimensional example: Let $\Rc$ be a vine copula with density $c$, where the following (unconditional) pair-copulas are specified:
\begin{table}[h]
	\centering
	\setlength{\tabcolsep}{5pt}
	\begin{tabular}{|l|cccccccc|}
		\hline
		pair-copula & $c_{1,2}$ & $c_{1,3}$ & $c_{3,4}$ & $c_{3,5}$ & $c_{2,6}$ & $c_{6,7}$ & $c_{7,8}$ & $c_{7,9}$\\
		\hline
		Kendall's $\tau$ & $-0.3$ & $0.5$ & $0.2$ & $-0.4$ & $0.5$ & $0.5$ & $-0.4$ & $0.6$\\
		\hline
	\end{tabular}
	\caption{Specification of the pair-copulas with empty conditioning set.}
	\label{tab:9dvine}
\end{table}\\
Now, we take an empty graph with node 1 to 9 and add edges $(i,j)$ if $c_{i,j}$ is specified in \autoref{tab:9dvine}. The result is a tree on the nodes 1 to 9 (see \autoref{fig:Foto}). We assign a ``$+$" to node 1 and consider its neighborhood $\left\lbrace 2,3\right\rbrace$ as there are still nodes without a sign. Since $\tau_{1,2}<0$ and the root node 1 has been assigned a ``$+$", node 2 gets a ``$-$". Node 3 is assigned a ``$+$". Next, we repeat this procedure for the neighborhoods of nodes 2 and 3. Iterating in this way until all nodes have been assigned a ``$+$" or a ``$-$" we obtain what is shown in \autoref{fig:Foto}. The resulting direction vector is given by $\vb=(1,-1,1,1,-1,-1,-1,1,-1)'$.
\begin{figure}
	\centering
	\includegraphics[width=.80\textwidth]{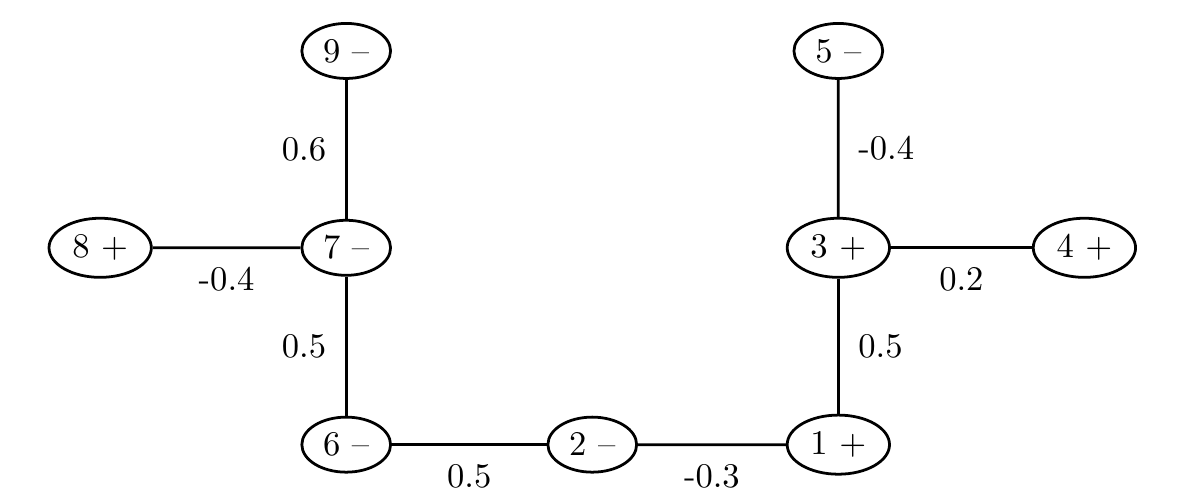}
	\caption{Example for finding the candidate vector.}
	\label{fig:Foto}
\end{figure}

\subsection{Procedure 2: Local search for better candidates}\label{App:better}
Having found a diagonal through Procedure 1 (\aref{App:start}), we additionally perform the following steps in order to look if there is a diagonal with even higher weight in the ``neighborhood'' of $\vb$.

\begin{enumerate}
	\item Consider a candidate diagonal vector $\vb\in\left\{1,-1\right\}^d$ with corresponding weight $\lambda_c^{(0)}$.
	\item For $j=1,\ldots,d$, calculate the weight $\lambda_c^{(j)}$ corresponding to $\vb_j\in\left\{1,-1\right\}^d$, where $\vb_j$ is equal to $\vb$ with the sign of the $j$th entry being reversed.
	\item If $\max_i \lambda_c^{(i)}>\lambda_c^{(0)}$, take $\vb:=\vb_k$ with $k=\argmax_i \lambda_c^{(i)}$ to be the new candidate for the (local) maximum.
	\item Repeat the steps 1--3 until a (local) maximum is found, i.e.\ $\max_i \lambda_c^{(i)}\leq\lambda_c^{(0)}$.
\end{enumerate}
Although there is no guarantee that we really find the global maximum of the diagonal weights, this procedure in any case finds a local maximum. Starting with a very plausible choice of $\vb$ it is highly likely that we end up with the ``right" diagonal.\\
In step 2 the weight of numerous diagonals has to be calculated. For a fast determination of these weights it is reasonable to approximate the integral in \autoref{eq:dwm} by
\[
\lambda_c(D)\approx \frac{1}{n}\sum_{i=1}^n c(\gammab (t_i))\left\|\dot{\gammab}(t_i)\right\|,
\]
where $0<t_1<t_2<\ldots<t_n<1$ is an equidistant discretization of $[0,1]$.

\bibliographystyle{apalike}
\bibliography{KLvines}

\end{document}